\documentclass[12pt, reqno]{amsart}
\makeatletter
\DeclareMathAlphabet{\mathpzc}{OT1}{pzc}{m}{it}
\usepackage[margin=2.1cm]{geometry}
\usepackage{amsmath, amssymb, amsfonts, amstext, verbatim, amsthm, mathrsfs}
\usepackage[mathcal]{eucal}
\usepackage{microtype}
\usepackage{epsfig}
\usepackage[all,arc,knot,poly]{xy}
\usepackage{graphicx}
\usepackage[usenames]{color}
\usepackage{aliascnt} 
\usepackage{tikz}
\usetikzlibrary{decorations.pathreplacing,cd}
\usetikzlibrary{matrix,arrows}
\usepackage[hidelinks]{hyperref}
\usepackage{enumerate}
\usepackage{xspace}
\usepackage{url}
\renewcommand{\omit}[1]{}

\theoremstyle{plain}

\newcommand{\refnewtheoremn}[4]{
\newaliascnt{#1}{#2}
\newtheorem{#1}[#1]{#3}
\aliascntresetthe{#1}
\expandafter\providecommand\csname #1autorefname\endcsname{#4}}

\newcommand{\refnewtheorem}[3]{\refnewtheoremn{#1}{#2}{#3}{#3}}

\def\makeCal#1{
\expandafter\newcommand\csname c#1\endcsname{\mathcal{#1}}}
\def\makeBB#1{
\expandafter\newcommand\csname b#1\endcsname{\mathbb{#1}}}
\def\makeFrak#1{
\expandafter\newcommand\csname f#1\endcsname{\mathfrak{#1}}}
\count@=0
\loop
\advance\count@ 1
\edef\y{\@Alph\count@}
\expandafter\makeCal\y
\expandafter\makeBB\y
\expandafter\makeFrak\y
\ifnum\count@<26
\repeat

\newtheorem{thm}{Theorem}[section]

\refnewtheorem{lemma}{thm}{Lemma}
\refnewtheorem{theorem}{thm}{Theorem}
\refnewtheorem{assumptions}{thm}{Assumptions}
\refnewtheorem{examples}{thm}{Examples}
\refnewtheorem{claim}{thm}{Claim}
\refnewtheorem{cor}{thm}{Corollary}
\refnewtheorem{conj}{thm}{Conjecture}
\refnewtheorem{prop}{thm}{Proposition}
\refnewtheorem{proposition}{thm}{Proposition}
\refnewtheorem{quest}{thm}{Question}
\refnewtheorem{assumption}{thm}{Assumption}
\refnewtheorem{problem}{thm}{Problem}
\refnewtheorem{corollary}{thm}{Corollary}
\refnewtheorem{conjecture}{thm}{Conjecture}

\refnewtheorem{remark}{thm}{Remark}
\refnewtheorem{remarks}{thm}{Remarks}
\refnewtheorem{defn}{thm}{Definition}
\refnewtheorem{definition}{thm}{Definition}
\refnewtheorem{notn}{thm}{Notation}
\refnewtheorem{const}{thm}{Construction}
\refnewtheorem{example}{thm}{Example}
\refnewtheorem{fact}{thm}{Fact}

\newcommand{\PGL}{\operatorname{PGL}}
\newcommand{\Quad}{\operatorname{Quad}}
\newcommand{\lra}{\longrightarrow}

\newcommand {\<}{\langle}
\renewcommand {\>}{\rangle}
\newcommand{\isom}{\cong}
\newcommand{\half}{\tfrac{1}{2}}
\newcommand{\tensor}{\otimes}
\newcommand{\mat}[4]{\begin{pmatrix}#1&#2\\#3&#4\end{pmatrix}}

\renewcommand{\Im}{\operatorname{Im}}

\newcommand{\dd}{\mathrm{d}}
\newcommand{\hk}{hyperk{\"a}hler }
\newcommand{\GL}{\operatorname{GL}}
\newcommand{\e}{\epsilon}
\newcommand{\hash}{\#}
\newcommand{\Res}{\operatorname{Res}}
\newcommand{\tr}{\operatorname{tr}}

\renewcommand{\dd}{d}

\begin{document}

\title{Joyce structures and poles of Painlev\'e equations}
\author{Tom Bridgeland}\author{Fabrizio Del Monte}

\date{}

\begin{abstract}
Joyce structures are a class of geometric structures that first arose in relation to Donaldson-Thomas theory. There is a special class of examples, called class $S[A_1]$, whose underlying manifold  parameterises Riemann surfaces of some fixed genus equipped with a meromorphic quadratic differential  with poles of fixed orders.  We study two Joyce structures of this type using the  isomonodromic systems associated to the Painlev{\'e} II and III$_3$ equations. We give explicit formulae for the  Pleba{\'n}ski functions of these Joyce structures, and compute several associated objects, including their tau functions, which we explicitly relate to the corresponding Painlev\'e tau functions. We show that the behaviour of the Joyce structure near the zero-section  can be studied analytically through poles of Painlev\'e equations. The systematic treatment gives a blueprint for the study of more general Joyce structures associated to meromorphic quadratic differentials on the Riemann sphere.

\end{abstract}

\maketitle

\section{Introduction}

This paper is concerned with a class of geometric structures known as Joyce structures.  These structures have appeared in several contexts recently, including integrable systems  \cite{A2,D,DM2} and  topological string theory \cite{AP,AP2}. A Joyce structure on a complex manifold $M$ involves a one-parameter family of flat and symplectic  non-linear connections on the tangent bundle $T_M$, and gives rise to a complex \hk structure on the total space $X=T_M$. The precise definition  first  appeared in \cite{RHDT2}, but the essential features are standard in  twistor theory (see e.g. \cite{DM}), and go back to work of  Pleba{\'n}ski \cite{P}. We give a brief summary in Section \ref{sum} below, and refer to \cite{JTw} for further details.

Joyce structures take their name from a line of research initiated in  \cite{holgen} which aims to encode the Donaldson-Thomas (DT)  invariants \cite{JS,KS1} of a three-dimensional Calabi-Yau (CY$_3$) category in a geometric structure on its space of stability conditions \cite{Stab1,Stab2}.   From this point of view a Joyce structure should be thought of as a non-linear analogue of a Frobenius structure \cite{Dub1,Dub2}, in which the linear structure group $\GL_n(\bC)$ is replaced by the group of symplectic automorphisms of a complex torus $(\bC^*)^n$. The wall-crossing formula shows that the DT invariants can be viewed as the Stokes data for an isomonodromic family of irregular connections on $\bP^1$ taking values in this group. This perspective is the subject of \cite{RHDT2}  and is summarised in \cite[Appendix A]{Strachan}.

An interesting  class of examples of Joyce structures was constructed in \cite{CH}.  The base $M$ parameterises  Riemann surfaces of some fixed genus $g$ equipped with a quadratic differential  with simple zeros.  The extension to spaces of quadratic differentials with poles of fixed orders $m=(m_1,\ldots,m_l)$ and zero residues will appear in \cite{Z}. These Joyce structures are expected to arise from the DT theory of the  CY$_3$ categories considered in \cite{BS, H}. We refer to them as Joyce structures of class $S[A_1]$ since they  are related  to supersymmetric gauge theories of this class.
There are ten examples  \cite{Bonelli2017} for which $g=0$ and the base $M$ has  dimension 2. We refer to these as the Painlev{\'e} Joyce structures, since they correspond to the ten $SL_2(\mathbb{C})$ isomonodromy problems associated to  the Painlev{\'e} equations in \cite{Chekhov2017,Saito2009}. The relevant pole  orders  are
\[\label{numbers}(2,2,2,2) \quad (2,2,4)\quad (2,2,3)\quad (6,2) \quad  (5,2) \quad  (4,4) \quad  (4,3) \quad  (3,3) \quad (8)\quad  (7).\]

The  case $m=(7)$ corresponding to Painlev{\'e} I was studied in detail in \cite{A2}.
The aim of this paper is to give explicit descriptions  in the cases $m=(3,3)$ corresponding to Painlev{\'e} III$_3$, and $m=(8)$ corresponding to Painlev{\'e} II. From the point of view of supersymmetric theories of class $S$ the case of Painlev\'e III$_3$ is particularly important, since it is related to the Seiberg-Witten theory of pure $SU(2)$ super Yang-Mills (see for example \cite{Bonell2017b,Gavrylenko2017}).  Both examples have several new features as compared with Painlev{\'e} I, and the relevant connections exhibit both local monodromy and Stokes phenomena. We give explicit formulae for the associated Pleba{\'n}ski functions, compute the linear Joyce connections, and investigate the Joyce structure tau function defined in \cite{JT}. The treatment is designed to be systematic, and to provide a blueprint for the study of more general Joyce structures associated to meromorphic quadratic differentials on the Riemann sphere.

\subsection*{Plan of the paper} We begin in Section \ref{sec:pain} with a quick introduction to Joyce structures as they arise in relation to the isomonodromy approach to Painlev{\'e} equations. In Section \ref{sum} we recall some of the general theory of Joyce structures and state our main results. In Section \ref{sec:Quad}  we outline a construction of Joyce structures of class $S[A_1]$  in terms of pencils of projective structures on Riemann surfaces, and in Section \ref{sec:Isomonodromy} we discuss an equivalent formalism using isomonodromic deformations of linear systems with a reference connection. Sections \ref{sec:three} and \ref{sec:two} are the heart of the paper and contain the  detailed computations in the two examples  corresponding to Painlev{\'e} III$_3$ and II   respectively. In Section \ref{taufun} we recall the definition of the Joyce structure tau function and compute it in our examples.

\subsection*{Acknowledgements} T.B. gratefully acknowledges unpublished joint work with Maciej Dunajski which included Maple calculations of the isomonodromic flows and the hyperk{\"a}hler structures for all the Painlev{\'e} Joyce structures. F.D.M. wishes to thank H. Desiraju, K. Iwaki and O. Lisovyy for useful discussions at various stages of this work.


\section{Painlev{\'e} Joyce structures}
\label{sec:pain}
We will give a detailed description of our main results in Section \ref{sum} below. Let us first try to quickly orient the reader who is familiar with the usual isomonodromy approach to Painlev{\'e} equations. For illustration purposes we will take the example of Painelv{\'e} I treated in \cite{A2}.

The standard isomonodromy approach to Painlev{\'e} I is to consider meromorphic connections on the trivial rank 2 bundle over $\bP^1$ of the form
\begin{equation}
    \label{form1}
\nabla=\dd-\mat{p}{x-q}{x^2+xq+q^2+t}{-p} \dd x.
\end{equation}
They have a single pole at $x=\infty$ and depend on three parameters $(t,q,p)$. As $t\in \bC$ varies we can uniquely vary $q=q(t)$ and $p=p(t)$ in such a way that the generalised monodromy (or Stokes data)  at $x=\infty$ remains constant. These isomonodromy flows can be put in Hamiltonian form by setting $H=p^2-q^3-tq$ and writing
\begin{equation}\frac{d q}{d t} = \frac{\partial H}{\partial p}=2p, \qquad \frac{d p}{d t} = -\frac{\partial H}{\partial q}=3q^2+t.\end{equation}
Eliminating $p$ leads to the Painlev{\'e} equation in the form
\begin{equation}
    \frac{d^2q}{dt^2}=2 \frac{dp}{dt}=6q^2+2t.
\end{equation}

Note that we can alternatively parameterise the connection by the variables $(t,H,q)$ with $p$ defined implicitly by $p^2=q^3+tq+H$. In these variables the  isomonodromic flow is described by the vector field
\begin{equation}
\label{aaa}
    \frac{\partial}{\partial t}-q\frac{\partial}{\partial H}+ 2p\frac{\partial}{\partial q}. 
\end{equation}

To define the associated Painlev{\'e} Joyce structure we need to change our point-of-view in several ways. Firstly, rather than a single connection $\nabla$, we must consider a one-parameter family or pencil of connections $\nabla_\epsilon$ depending on an additional parameter $\epsilon\in \bC^*$
\begin{equation}
    \label{form2}
\nabla_\epsilon=d-\mat{r}{0}{0}{-r} dx-\frac{1}{\epsilon} \mat{p}{x-q}{x^2+xq+q^2+t}{-p} dx.
\end{equation}
The formula  \eqref{form2} is obtained from \eqref{form1} by introducing $\epsilon$ in the obvious way, and then performing a simple shift\footnote{To achieve uniformity with the rest of this paper we are using slightly different notation to \cite{A2}. The variables $(r,p)$ there correspond to $(2pr,-p)$ here.}  $p\mapsto p+\epsilon r$. Thus for a fixed value of $\epsilon\in \bC^*$ the connection $\nabla_\epsilon$ depends only on three parameters. Nonetheless, the pencil of connections \eqref{form2} depends on four parameters, with $(t,q,p)$ specifying the Higgs field part multiplying $\epsilon^{-1}$, and the additional parameter $r$ specifying the reference connection $\nabla_\infty$.

Writing $H=p^2-q^3-tq$ as before we can equivalently view the pencil of connections \eqref{form2} as being parameterised by the co-ordinates $(t,H,q,r)$ with $p$ being once again defined implicitly by $p^2=q^3+tq+H$. For a fixed $\epsilon\in \bC^*$ we can now consider varying both $t$ and $H$  and allowing $q=q(t,H)$ and $r=r(t,H)$ to vary in such a way that the generalised monodromy of  $\nabla_\epsilon$ remains constant. This leads to a two-dimensional space of flows spanned by the vector fields

\begin{gather}
\label{seven}
\frac{2p}{\epsilon} \frac{\partial}{\partial q}-\frac{(3q^2+t)r}{\epsilon p} \frac{\partial}{\partial r}+\frac{\partial}{\partial t}-q\frac{\partial}{\partial H}+2r\frac{\partial}{\partial q}, \\
-\frac{1}{2p\epsilon}\frac{\partial}{\partial r}+\frac{\partial}{\partial H} .
\end{gather}
Note that the first of these flows reproduces \eqref{aaa} when $\epsilon=1$ and $r=0$. Moreover, whatever the value of $r$, the resulting differential equation for $q$  is always the $\e$-deformed Painlev\'e I equation
\begin{equation}
    \frac{d^2q}{dt^2}=\frac{1}{\e^2}(6q^2+2t).
\end{equation}

In general, the flows \eqref{seven} seem to be  new information, and lead to previously unknown structures including a complex \hk structure on the four-dimensional complex manifold  $X^\hash$  whose co-ordinates are $(t,H,q,r)$. 
The final step in the description of a Joyce structure is to re-express these flows in a different set of co-ordinates on the space $X^\hash$.

The connection $\nabla_\epsilon$ can be transformed by a singular gauge transformation into an oper with an apparent singularity expressed by the equation
\begin{equation}
    y''(x)=Q(x)y(x), \qquad Q(x)=\e^{-2}\cdot Q_0(x)+\epsilon^{-1}\cdot Q_1(x)+Q_2(x)
\end{equation}
\begin{equation}Q_0(x) = x^3+tx+H,\qquad Q_1(x)=-\frac{p}{x-q}+2pr, \qquad Q_2(x)=\frac{3}{4(x-q)^2}-\frac{r}{x-q}+r^2.\end{equation}
The required co-ordinates $(z_1,z_2,\theta_1,\theta_2)$ are then defined   by integrals of the form
\begin{equation}z_i=\oint_{\gamma_i}\sqrt{Q_0} \, dx, \qquad \theta_i=-\oint_{\gamma_i} \frac{Q_1}{2\sqrt{Q_0}} \, dx,\end{equation}
where $\gamma_1,\gamma_2$ are   cycles on the spectral curve $\Sigma=\{(x,y): y^2=Q_0(x)\}$.
One way to motivate these co-ordinates is via WKB analysis, which predicts that the logarithmic Fock-Goncharov co-ordinates $x_i$ of the  connection $\nabla_\epsilon$ have asymptotics \begin{equation}x_i(\epsilon)=-\frac{z_i}{\epsilon} + \theta_i + O(\epsilon).\end{equation}

In these co-ordinates the two-dimensional space of flows is spanned by the vector fields
\begin{equation}
\label{above'}\frac{\partial}{\partial z_i} + \frac{1}{\epsilon} \cdot \frac{\partial}{\partial \theta_i}+ 2\pi i\cdot  \frac{\partial^2 W}{\partial \theta_i\partial \theta_1} \cdot \frac{\partial}{\partial \theta_2}- 2\pi i \cdot \frac{\partial^2 W}{\partial \theta_i\partial \theta_2} \cdot \frac{\partial}{\partial \theta_1},\end{equation}
where $W=W(z_i,\theta_j)$ is a single function called the Pleba{\'n}ski function of the Joyce structure. It was shown in \cite{A2} that
\begin{equation} 
W= \frac{p}{2(4t^3+27H^2)} \left(t- (9H-6tq)r+ (8t^2-18Hq+12tq^2)r^2+8tp^2r^3\right).\end{equation}
We will derive similar expressions in the   Painlev{\'e} III$_3$ and Painlev{\'e} II examples below.


\section{Statement of results}
\label{sum}

In this section we first briefly review some required background material on Joyce structures, referring the reader to \cite{JTw,JT} for more details. We  then state our main results.

\subsection{Joyce structures and Pleba{\'n}ski functions}

\label{early}

Let $M$ be a complex manifold of dimension $n$. A Joyce structure on  $M$ involves a system of $n$  flows on the total space of the tangent bundle $X=T_M$ of the form
\begin{equation}
\label{above}\frac{\partial}{\partial z_i} + \frac{1}{\epsilon} \cdot \frac{\partial}{\partial \theta_i}+ \sum_{p,q} \eta_{pq} \cdot \frac{\partial^2 W}{\partial \theta_i\partial \theta_p} \cdot \frac{\partial}{\partial \theta_q} .\end{equation}
Here $(z_1,\ldots,z_n)$ are a canonical system of co-ordinates on $M$ uniquely defined  up to integral linear transformations $z_i\mapsto \sum_j a_{ij} z_j$ with $(a_{ij})_{i,j=1}^n\in \GL_n(\bZ)$. These induce linear co-ordinates $(\theta_1,\ldots,\theta_n)$ on the tangent spaces $T_p M$ by writing a tangent vector  in the form  $\sum_i \theta_i\cdot  \frac{\partial}{\partial z_i}\big|_p$, and we can then consider $(z_i,\theta_j)$ as co-ordinates on the total space $X=T_M$.   Finally, $(\eta_{pq})_{p,q=1}^n$ is an invertible, skew-symmetric matrix, and $\epsilon\in \bC^*$ is an additional parameter.

The flows \eqref{above} are specified by a single function $W=W(z_i,\theta_j)$ on $X$ called the Pleba{\'n}ski function. This is required to satisfy Pleba{\'n}ski's second heavenly equations 
\begin{equation}\label{plebeq}\frac{\partial^2 W}{\partial \theta_i \partial z_j}-\frac{\partial^2 W}{\partial \theta_j \partial z_i }=\sum_{p,q} \eta_{pq} \cdot \frac{\partial^2 W}{\partial \theta_i \partial \theta_p} \cdot \frac{\partial^2 W}{\partial \theta_j \partial \theta_q},\end{equation}
which ensures that the  flows  \eqref{above} are compatible for any $\epsilon\in \bC^*$. We also impose the following additional conditions:

\begin{gather}
\label{go2}
\frac{\partial^2 W}{\partial \theta_j \partial \theta_k}(z_1,\ldots,z_n,\theta_1+2\pi i k_1,\ldots,\theta_n+2\pi ik_n)=\frac{\partial^2 W}{\partial \theta_j \partial \theta_k}(z_1,\ldots,z_n,\theta_1,\ldots, \theta_n), \\
  \label{re2}W(\lambda z_1,\ldots,\lambda z_n,\theta_1,\ldots, \theta_n)=\lambda^{-1} \cdot W(z_1,\ldots,z_n,\theta_1,\ldots, \theta_n), \\
   \label{st2}
W(z_1,\ldots,z_n,-\theta_1,\ldots, -\theta_n)=-W(z_1,\ldots,z_n,\theta_1,\ldots, \theta_n),\end{gather}
where   $(k_1,\ldots, k_n)\in \bZ^n$ in \eqref{go2}, and $\lambda\in \bC^*$ in \eqref{re2}.

Note that the flows \eqref{above} only define the Pleba{\'n}ski function $W=W(z_i,\theta_j)$ up to the addition of functions of the form $a(z_i)+\sum_j b_j(z_i) \theta_j$. In the two examples  studied in this paper there is in fact a unique choice for $W$ which satisfies the stronger version of \eqref{go2}
\begin{equation}
    \label{per}
W(z_1,\ldots,z_n,\theta_1+2\pi i k_1,\ldots,\theta_n+2\pi ik_n)=W(z_1,\ldots,z_n,\theta_1,\ldots, \theta_n),\end{equation}
and in what follows we will always take this choice. 
Note that  $W$ then descends to a function on the quotient $X^{\hash}=X\, / \,(2\pi i\, \bZ)^n$
whose local co-ordinates are $(z_i,e^{\theta_j})$. 
For each of the two examples considered in this paper we will give an explicit rational expression  for  $W$ in terms  of  a different system of  co-ordinates on $X^\hash$  which are natural for the associated isomonodromy problem.

\subsection{Restriction to the zero section}
\label{summy}
The two examples of Joyce structures  appearing in this paper are meromorphic, in the sense that the Pleba{\'n}ski function $W$ is a meromorphic function on $X$. It is then a subtle problem to determine the positions of the poles of $W$.  In both examples we  will establish the non-trivial fact that   $W$ is regular along the zero section $M\subset X=T_M$  defined by setting all $\theta_i=0$. 

Once this regularity is established, it follows from  \eqref{plebeq} and \eqref{st2}  that there is a locally-defined function $S=S(z_i)$ on $M$,  well-defined up to the addtion of a constant,  such that
\begin{equation}
    \label{s}
\frac{\partial W}{\partial \theta_i}\bigg|_{\theta=0} =\frac{\partial S }{\partial z_i}.\end{equation}
In both our examples (and also in the Painlev{\'e} I case), the function $S$  has the  form
	\begin{equation}\label{hi2}
		S=\log\Delta^{-\frac{1}{24}},
	\end{equation}
   where $\Delta$ is closely related to the discriminant of the spectral curve $y^2=Q_0(x)$.  This suggests a possible link with  the Bergman tau function \cite{Korotkin2020} which, in the case of a genus 1 curve, is the square of the Dedekind eta function.

Another reason that the regularity of $W$ along the zero section is interesting is that, following \cite[Section 7]{RHDT2},  it allows us to define a flat, torsion-free connection $\nabla^J$  on the tangent bundle $T_M$ called the linear Joyce connection. This connection is defined by the formula
\begin{equation}
\label{st}\nabla^J_{\frac{\partial}{\partial z_i}}\Big(\frac{\partial}{\partial z_j}\Big)=  -\sum_{p,q}\eta_{pq}\cdot\frac{\partial^3 W}{\partial \theta_i \, \partial \theta_j \, \partial \theta_p}\bigg|_{\theta=0} \cdot \frac{\partial}{\partial z_q}.\end{equation}
We compute this connection in our two examples, and express the answer by giving a system of  flat co-ordinates on $M$.

\subsection{Tau functions}
Joyce structures arising from DT theory seem to have interesting links to non-perturbative completions of topological string partition functions. The definition of the Joyce structure tau function was introduced in \cite{JT} in an attempt to better understand this relation. The definition depends on the choice of certain extra data and is rather experimental in nature. It is therefore of interest to compute it in our examples.

A Joyce structure on a complex manifold $M$ defines a complex \hk structure on $X=T_M$, and associated closed holomorphic 2-forms $\Omega_\pm=\Omega_{J\pm iK}$ and $\Omega_I$ on $X$.
For each $\epsilon\in \bC^*$ the combination
\begin{equation}
    \Omega_\epsilon=\epsilon^{-2}\Omega_++2i\epsilon^{-1} \Omega_I+\Omega_\infty
\end{equation}
descends to a symplectic form on the twistor fibre $Z_\epsilon$, which is the space of leaves of the flows \eqref{above}. If we set $\Omega_0=\Omega_+$ and $\Omega_\infty=\Omega_-$ this statement holds for all $\epsilon\in \bP^1$ after appropriate rescaling.

Let us  choose primitives
\begin{equation}\label{pot1}d\Theta_0=\Omega_0, \qquad d\Theta_{1}=\Omega_1,\qquad d\Theta_\infty=\Omega_\infty,\qquad d\Theta_I=\Omega_I.\end{equation}
We then define the Joyce structure tau function locally on $X$  by the relation
\begin{equation}
\label{tau1}d\log(\tau)=\Theta_0+2i\Theta_I +\Theta_\infty-\Theta_1.\end{equation}

This definition is  of course vacuous without some procedure for defining the primitives \eqref{pot1}. At least for $\Theta_0,\Theta_1$ and $\Theta_I$ there are natural choices which will be explained in Section \ref{tauf}. The choice of  $\Theta_\infty$ is more mysterious, but in our two examples this problem can be side-stepped by restricting to a natural Lagrangian submanifold of $Z_\infty$ and taking $\Theta_\infty=0$.  After doing this we find that the Joyce structure tau function gives a particular normalisation of the  corresponding Painlev{\'e} tau function.

\subsection{Painlev{\'e} {III$_3$} case}

This case is treated in detail in Section \ref{sec:three}. 
The base $M$ of the Joyce structure parameterises quadratic differentials
\begin{equation}Q_0(x)  \dd x^{\tensor 2}= \frac{1}{x^2}\left(tx+H +x^{-1}\right)   \dd x^{\tensor 2},\end{equation}
on $\bP^1$ which have poles  at $x=0$ and $x=\infty$ of order 3, and simple zeros. The quotient space $X^\hash$  parameterises pencils of connections\footnote{Strictly speaking, up to the involution \eqref{inv1}; see Lemma \ref{abhol}.} of the form
\begin{equation}
    \label{con}
\nabla_\epsilon=\dd-  \left(\begin{array}{cc}
		r & 0 \\
		0 & -r
	\end{array} \right)\frac{\dd x}{x}-\frac{1}{\e}\left(\begin{array}{cc}
		pq & 1-qx^{-1} \\
		 tx-q^{-1}  & -pq
	\end{array} \right)\frac{\dd x}{x}, \end{equation}
    where $p$ is defined implicitly by $p^2=Q_0(q)$. 
    This connection $\nabla_\epsilon$ is equivalent to an oper with apparent singularity  defined by
    \begin{equation}
    \label{q}
y''(x)=Q(x)y(x), \qquad Q(x)=\e^{-2}Q_0(x)+\epsilon^{-1}Q_1(x)+Q_2(x)\end{equation}
where 
 \begin{equation}Q_1(x)=-\frac{pq^2}{x^2(x-q)}+\frac{2pqr}{x^2}, \qquad
 Q_2(x)=\frac{3}{4(x-q)^2}-\frac{x+rq}{x^2(x-q)}+\frac{r^2}{x^2}.
 \end{equation}
The canonical co-ordinates are given by  integrals of the form
\begin{equation}\label{zt}z_i=\oint_{\gamma_i}\sqrt{Q_0} \, dx, \qquad \theta_i=-\oint_{\gamma_i} \frac{Q_1}{2\sqrt{Q_0}} \, dx,\end{equation}
where $\gamma_1,\gamma_2$ are   cycles on the spectral curve $\Sigma=\{(x,y): y^2=Q_0(x)\}$.

The flows of the Joyce structure \eqref{above} are obtained by allowing $(t,H)$ to vary and insisting that $q=q(t,H)$ and $r=r(t,H)$ change in such a way that the generalised monodromy of the connection \eqref{con} is constant.  
By computing these flows and re-expressing them in the co-ordinates $(z_i,\theta_j)$ we prove the following result.
 
\begin{thm}\begin{itemize}
    \item [(i)]
The Pleba{\'n}ski function on $X^\hash$ is given by the formula
\begin{equation}W=\frac{pq}{6(H^2-4t)}\big(tq+(H+6tq)r+(6H+12tq)r^2+8p^2q^2r^3\big).\end{equation}
\item[(ii)] The Pleba{\'n}ski function is regular on the locus $\theta=0$ and
\begin{equation}S=\log \big(H^2-4t\big)^{-\frac{1}{24}}.\end{equation}
\item [(iii)] The local co-ordinates $(s=\log(t),H)$ on $M$ are flat for the linear Joyce connection.

\item[(iv)] With the natural choices for the primitives \eqref{pot1} we have
\begin{equation}d\log(\tau|_{Y^{\hash}})=- H d s+ p d q  + d(4H+2qp)+\frac{1}{4\pi i} \,x_1 dx_2,\end{equation}
where $Y^\hash\subset X^\hash$ is the locus $r=0$, and $(x_1,x_2)$ are  logarithms of Fock-Goncharov co-ordinates of the connection \eqref{con}.
\end{itemize}
\end{thm}

One special feature of this case is an involution
\begin{equation}\label{inv1}
         (q,p,r)\mapsto \left((qt)^{-1},-tpq^2 ,-(r+\half)\right).
    \end{equation}
    which preserves all objects of interest, including the Pleba{\'n}ski function $W$.
    
\subsection{Painlev{\'e} {II} case}

This case is the subject of Section \ref{sec:two}. The base $M$ of the Joyce structure parameterises quadratic differentials
\begin{equation}Q_0(x) \, \dd x^{\tensor 2}=(x^4+tx^2+2H)\,\dd x^{\tensor 2},\end{equation}
with a single pole of order 8 at $x=\infty$ with zero residue, and simple zeros.
The quotient space $X^\hash$  parameterises pencils of connections of the form
\begin{equation}
    \label{con21}
\nabla_\epsilon=\dd-  \left(\begin{array}{cc}
		r & 0 \\
		-2r(x+q)-1 & -r
	\end{array} \right)\dd x-\frac{1}{\e}\left(\begin{array}{cc}
		x^2+p-q^2 & x-q \\
		(x+q)(t-2p+2q^2)   & -x^2-p+q^2
	\end{array} \right)\dd x. \end{equation}
    The connection $\nabla_\epsilon$ can again be put in the oper form \eqref{q} with
\begin{equation}
	Q_1(x)=-\frac{p}{x-q}
	+2pr,\qquad 
	Q_2(x)=\frac{3}{4(x-q)^2}-\frac{r}{x-q}+r^2.
\end{equation}
The canonical co-ordinates $(z_i,\theta_j)$ are defined by integrals of the form \eqref{zt} as before, and computing the isomonodromic flows in these co-ordinates gives the following result.

\begin{thm}\label{Ptwothm}\begin{itemize}
    \item [(i)]
The Pleba{\'n}ski function on $X^\hash$ is given by the formula
\begin{equation} W=\frac{p}{48H(t^2-8H)}\big(-tq - 2r(2t^2+3q^2t-12H) +12r^2q(-t^2-q^2t+4H)-8r^3p^2t\big).\end{equation}
\item[(ii)] The Pleba{\'n}ski function is regular on the locus $\theta=0$ and
\begin{equation}S=\log \big(H^2(8H-t^2)\big)^{-\frac{1}{48}}.\end{equation}
\item [(iii)] The local co-ordinates $\big(t,H-\tfrac{1}{8}t^2\big)$ on $M$ are flat for the linear Joyce connection.
\item[(iv)] With the natural choices for the primitives \eqref{pot1} we have
\begin{equation}\label{hallam}d\log(\tau|_{Y^\hash})=
- H d t+ p d q  + \frac{1}{3}d(2tH-qp)+\frac{1}{2\pi i}\, x_1 dx_2,\end{equation}
where $Y^\hash\subset X^\hash$ is the locus $r=0$, and $(x_1,x_2)$ are  logarithms of Fock-Goncharov co-ordinates of the connection \eqref{con}.
\end{itemize}
\end{thm}
After matching parameters the expression appearing on the right-hand side of \eqref{hallam} coincides (up to a factor of two) with the normalisation of the Painlev{\'e} II tau function considered by Its, Lisovyy and Prokhorov \cite{Its2016}.


\section{Approach via opers with apparent singularities}
\label{sec:Quad}

As discussed in the introduction there is a class of meromorphic Joyce structures which are related to supersymmetric gauge theories of class $S[A_1]$.  The base $M$  parameterises Riemann surfaces equipped with quadratic differentials  having poles of fixed orders with zero residues, and simple zeros. In the case  of  quadratic differentials without poles these Joyce structure were constructed in \cite{CH}. The rigorous construction in the meromorphic case  will appear in \cite{Z}. In this section we give a sketch of a general but conjectural  approach to these Joyce structures following \cite[Section 4]{JT}. This will  be worked out in detail in our two examples in Sections \ref{sec:three} and \ref{sec:two} below. Another interesting class of examples involving quadratic differentials on the Riemann sphere with a single pole of odd order is studied in \cite{DM2}.

\subsection{Quadratic differentials}

We begin by fixing the data of a genus $g\geq 0$ and a  collection of integers  $m=(m_1,\ldots,m_l)$ with all $m_i\geq 2$. We always assume that $l>0$ and
\begin{equation}k:=6g-6+\sum_{i=1}^l (m_i+1)>0.\end{equation}
We consider the moduli space
$\Quad(g,m)$ parameterizing pairs $(C,Q_0)$, where $C$ is a compact, connected Riemann surface of genus $g$, and $Q_0=Q_0(x)\, \dd x^{\tensor 2}$ is a meromorphic section of $\omega_C^{\tensor 2}$ with simple zeros, and 
  $l$ unordered poles $x_i\in C$ with multiplicities $m_i$. This space is a complex orbifold of dimension $k$.  We now give a brief review of its  basic properties, referring the reader to \cite[Sections 2 - 4]{BS} for more details.

Given a point $(C,Q_0)\in \Quad(g,m)$ there is a spectral curve $p\colon \Sigma\to C$ branched at the zeros and odd order poles of $Q_0$.  Locally it is given by writing $y^2=Q_0(x)$. There is a covering involution $\sigma\colon \Sigma\to \Sigma$ and a canonical meromorphic 1-form $\lambda=y \,\dd x$ satisfying $\lambda^{\tensor 2}=p^*(Q_0)$.
We define $\Sigma^0\subset \Sigma$ to be the complement of the poles of $\lambda$.

We consider the homology group $H_1(\Sigma^0,\bZ)^-$. The superscript signifies anti-invariance for the covering involution: we consider only classes satisfying $\sigma_*(\gamma)=-\gamma$. A calculation shows that $H_1(\Sigma^0,\bZ)^-\isom \bZ^{\oplus k}$ is free of  rank $k$. 
Given  a basis $(\gamma_1,\ldots,\gamma_k)\subset H_1(\Sigma^0,\bZ)^-$ we define
\begin{equation}\label{zzz}z_i=\oint_{\gamma_i} \lambda\in \bC.\end{equation}

As the point $(C,Q_0)\in \Quad(g,m)$ varies, the homology groups $H_1(\Sigma^0,\bZ)^-$ fit together to form a  local system over $\Quad(g,m)$. Transporting the basis elements $\gamma_i$ to nearby points, the resulting functions $(z_1,\ldots,z_k)$ form local co-ordinates on $\Quad(g,m)$. In particular, the tangent space at a point of $\Quad(g,m)$ is naturally identified with the cohomology group $H^1(\Sigma^0,\bC)^-$.

The inclusion $i\colon \Sigma^0\hookrightarrow \Sigma$ defines a natural map
\begin{equation}
    \label{mama}
i_*\colon H_1(\Sigma^0,\bZ)^-\to H_1(\Sigma,\bZ)^-\end{equation}
whose image  is a finite index subgroup $\Gamma\subset H_1(\Sigma,\bZ)^-$. The  intersection form  on $H_1(\Sigma,\bZ)$ pulls back via $i_*$ to  an integral  skew-symmetric form $\<-,-\>$ on $H_1(\Sigma^0,\bZ)^-$.
The form $\<-,-\>$ then  induces a Poisson structure on $\Quad(g,m)$ satisfying
\begin{equation}\label{poisson}\{z_i,z_j\}=2\pi i \,\<\gamma_i,\gamma_j\>.\end{equation}

After tensoring with $\bQ$, the kernel of  the map \eqref{mama} is spanned by classes $\pm \beta_i\in H_1(\Sigma^0,\bZ)^-$  defined up to sign by the difference of small loops  around the two inverse images of an even-order pole  of $Q_0$.
The residue of $Q_0$ at such a  pole $x_i\in C$ is defined to be the period
\begin{equation}\label{res2}\Res_{x_i}(Q_0)=\pm \oint_{\beta_i}\lambda\in \bC.\end{equation}
It is well-defined up to sign and is a Casimir for the above Poisson structure. Fixing these residues locally cuts out a symplectic leaf  in $\Quad(g,m)$ of  dimension $n=2d$, where
\begin{equation}\label{d} d=3g-3+ \sum_{i=1}^l \, \bigg\lceil \frac{m_i}{2} \bigg\rceil.\end{equation}
The tangent space to this leaf  is identified with  the cohomology group $H^1(\Sigma,\bC)^-$. 

The base of our Joyce structure will be the particular symplectic leaf
\begin{equation}\label{m}M=M(g,m)\subset \Quad(g,m)\end{equation}
obtained by requiring all residues \eqref{res2} at  even order poles of $Q_0$ to be zero. Note that a pole of $Q_0$  of order 2 with zero residue is in fact a pole of order 1.

Choosing a basis $(\gamma_1,\ldots, \gamma_n)$ for the finite index subgroup $\Gamma\subset H_1(\Sigma,\bZ)^-$ gives preferred local co-ordinates $(z_1,\ldots,z_n)$ on $M$ defined by the formula \eqref{zzz}.
As in Section \ref{early}  there are corresponding linear co-ordinates $(\theta_1,\ldots,\theta_n)$ on the tangent spaces $T_p M$ obtained by writing a tangent vector  in the form  $\sum_i \theta_i\cdot  \frac{\partial}{\partial z_i}\big|_p$. We  can then consider $(z_i,\theta_j)$ as co-ordinates on the total space $X=T_M$. We also introduce the quotient $X^{\hash}=X/ (2\pi i)\, \bZ^n$
whose local co-ordinates are $(z_i,e^{\theta_j})$.

\subsection{Isomonodromy flows}
\label{sec:opers}

To define a Joyce structure on $M$ we need a system of flows \eqref{above}. 
We will construct  these by associating to a generic  point of the space $X^\hash$  a  family of differential  equations of the form
\begin{equation}\label{above2}y''(x)=Q(x,\epsilon) \cdot y(x), \qquad Q(x,\epsilon)=\epsilon^{-2}\cdot {Q_0(x)}+\epsilon^{-1}\cdot Q_1(x) +  Q_2(x),\end{equation}
depending on a parameter  $\epsilon \in \bC^*$. In concrete terms the $Q_i(x)$ are meromorphic functions on $C$, but more  invariantly,  $Q_1= Q_1(x) \, \dd x^{\tensor 2}$ is another meromorphic quadratic differential on $C$, and $Q_2$ represents  a meromorphic projective structure \cite{AB}.
The required flows \eqref{above} will then be given by allowing the point $(C,Q_0)$ to move in the space $M$ and varying the other data $Q_1,Q_2$ so that the generalised monodromy of the equation \eqref{above2} remains constant.

At a point $x_i\in C$ where $Q_0$ has a pole of order $m_i$ we will insist that $Q_1$ and $Q_2$ have poles of order at most  $\lceil \half m_i\rceil$. We also allow $Q_1$ and $Q_2$ to have poles at exactly $d$ other points $q_i\in C$, with $d$ given by \eqref{d}. At these points  $Q_1,Q_2$  will be required to have the leading-order  behaviour
\begin{equation}\label{eq:Q1Q2app}
Q_1(x)= -\frac{p_i}{x-q_i} + u_i + O(x-q_i), \qquad Q_2(x)=\frac{3}{4(x-q_i)^2} - \frac{r_i}{x-q_i} + v_i +O(x-q_i),\end{equation}
with $p_i,r_i,u_i,v_i\in \bC$.
We will then insist that for all $\epsilon\in \bC^*$ the equation \eqref{above2}  has apparent singularities at the points  $x=q_i$, i.e. that the monodromy of the associated linear system is trivial as an element of $\PGL_2(\bC)$. 
This is equivalent to the condition 
\begin{equation} \label{eq:apparent}
	\left(\epsilon^{-1} p_i+r_i\right)^2=\epsilon^{-2}Q_0(q_i) +\epsilon^{-1} u_i+v_i,
\end{equation}
 for all $\epsilon\in \bC^*$ (see e.g. \cite[Lemma 2.1]{A2}), and hence to the equations \begin{equation}
\label{conditions}p_i^2=Q_0(q_i), \qquad u_i=2p_ir_i, \qquad v_i=r_i^2.\end{equation} 
The first of these relations shows that  the pair $(q_i,p_i)$ defines a point of the spectral curve  $\Sigma$.

The condition on the pole orders of $Q_1$ at the points $x_i$, together with the equations \eqref{conditions},  ensures that  the anti-invariant differential \begin{equation}\theta=-\frac{Q_1(x)\, dx}{2\sqrt{Q_0(x)}}\end{equation} on $\Sigma$ has simple poles at the points $(q_i,\pm p_i)$ with residues $\pm \half$, and no other poles. We will then require that if the equation \eqref{above2} is associated to a point of $X^\hash$ with co-ordinates $(z_i,e^{\theta_j})$ then
\begin{equation}
    \label{xi2}
e^{\theta_j}=\exp\left(\oint_{\gamma_j} \theta\right).\end{equation}
Note that the right-hand side of \eqref{xi2} is  well-defined for $\gamma_i\in H_1(\Sigma^0,\bZ)^-$. Indeed, we can lift $\gamma_i$ to an anti-invariant cycle on the curve $\Sigma^0$ further punctured at the points $(q_i,\pm p_i)$. Any two ways of doing this differ by a sum of cycles of the form $\delta_i - \sigma^*(\delta_i)$, where $\delta_i$ is a small loop around the point $(q_i,p_i)$. By the residue property of the differential $\theta$ such a  change affects the integral on the right-hand side of \eqref{xi2} by an element of $2\pi i \bZ$.

It is conjectured in \cite[Section 4]{JT} that for a generic point of $X^\hash$ lying over a point $(C,Q_0)\in M$, there is a unique equation \eqref{above2} satisfying the conditions discussed above and \eqref{xi2}. In other words, with the given conditions  on the poles of $Q_1$ and $Q_2$, the equation \eqref{xi2} determines $Q_1$ uniquely, and then generically $Q_2$ is also uniquely determined by the conditions \eqref{conditions} coming from the assumption that \eqref{above2} has apparent singularities.
It is further claimed in  \cite[Section 4]{JT} that the isomonodromy flows for the equation \eqref{above2} give rise to a Joyce structure on $M$.

A general construction of these Joyce structures  will be given in \cite{Z} with a different  approach involving bundles with connection rather than projective structures with apparent singularities. 
In this paper we will construct the required Joyce structures in our two examples by direct calculation. This involves computing the isomonodromy flows for the equation \eqref{above2} and rewriting them in the co-ordinates $(z_i,\theta_j)$ given by \eqref{zzz} and \eqref{xi2}.

\begin{remark}
\label{earlyish}
There is a small subtlety involved with interpreting \eqref{xi2} in the case when $Q_0$ has even order poles. 
Restricting to the symplectic leaf \eqref{m} involves setting the co-ordinates $z_i$ defined by the cycles $\pm \beta_i$ equal to $0$, and the corresponding co-ordinates $e^{\theta_i}$ equal to 1. However, the classes $\pm \beta_i$ only span an index 2 subgroup of the kernel of the map \eqref{mama}: there is a further generator $\nu\in \ker(i_*)\subset H_1(\Sigma^0,\bZ)^-$, well-defined up to the addition of linear combinations of the $\pm \beta_i$,  given by the sum of small anti-clockwise cycles around one inverse image of each even order pole of $Q_0$. To specify the Joyce structure on $M$  we must specify the value of the expression $e^{\theta_\nu}$ corresponding to this generator. Since $2\nu$ is a linear combination of the classes $\pm \beta_i$ we necessarily have $e^{\theta_\nu}\in \{\pm 1\}$. In the Painlev{\'e} II example below we will take $e^{\theta_\nu}=-1$ since this relates most directly to the existing literature on Painlev{\'e} tau functions.
\end{remark}


\section{Approach via linear systems}
\label{sec:Isomonodromy}

In this section, we will recall how to describe in a systematic manner isomonodromic deformations of general equations of the form \eqref{above2}, which are given by Painlev\'e equations when $C=\bP^1$ and the spectral curve $\Sigma$ has genus 1. Rather than considering scalar ODEs as in the previous section, the modern formulation of isomonodromic deformations and Painlev\'e equations typically involves matrix-valued linear systems of ODEs, describing deformations of flat $\mathfrak{g}$-connections (see for example the monograph \cite{fokas2006painleve}). Our description will differ from the usual integrable system literature in that we will always work with a nontrivial reference connection. This will mean that there are twice as many isomonodromic flows as in the usual Painlev\'e literature, albeit half of them take a very simple form.

In this paper we will only consider $\mathfrak{g}=\mathfrak{sl}_2$ and $C=\mathbb{P}^1$. The extension to higher genus Riemann surfaces is complicated by the existence of nonequivalent flat bundles. While it is certainly possible to formulate isomonodromic deformations also in this  case \cite{Levin1999,Krichever2001}, analytic control has only been achieved recently in the case of genus 1 \cite{DMDG2020,DMDG2022}.  

\subsection{Isomonodromic deformations with a reference connection}

Consider a meromorphic flat $\mathfrak{g}$-connection on a rank 2 trivial bundle on $\mathbb{P}^1$ of the form
\begin{equation}\label{eq:}
    \nabla_\epsilon:=\dd-A_0(x)\dd x-\frac{1}{\epsilon}\Phi(x)\dd x.
\end{equation}
 Here $\dd-A_0(x)\dd x$ is a reference connection, and $\Phi(x)\dd x$ is a Higgs field, namely a $\mathfrak{g}$-valued meromorphic 1-form, and $\epsilon \in \bC^*$.  We can relate this to an oper with apparent singularities  by performing  a singular gauge transformation
\begin{equation}
\label{gauge} g(x)^{-1} \big(\dd-  A(x)  \dd x\big) g(x) = \dd- \left( \begin{array}{cc}0 & 1 \\ Q(x) & 0 \end{array}\right) \dd x.
\end{equation}
This can be achieved by taking
    \begin{equation}
g(x) = \left( \begin{array}{cc}A_{12}(x)^{\frac{1}{2}} & 0 \\ 0 & A_{12}(x)^{-\frac{1}{2}}\end{array} \right) \cdot \left( \begin{array}{cc}
    1 & 0 \\ \frac{A'_{12}(x)}{2A_{12}(x)}-A_{11}(x)  & 1
\end{array} \right).
    \end{equation}  
The resulting potential is
\begin{equation}\label{az}
        Q(x)=-\det A(x) +A_{11}(x) \bigg(\frac{A'_{11}(x)}{A_{11}(x)}-\frac{A'_{12}(x)}{A_{12}(x)}\bigg)+\frac{3}{4}\bigg(\frac{A'_{12}(x)}{A_{12}(x)}\bigg)^2-\frac{A''_{12}(x)}{2A_{12}(x)},
    \end{equation}
and the condition $\nabla_\epsilon\Psi=0$ for a flat section becomes the second order ODE \eqref{above2} after setting $\Psi=(y(x),  y'(x))^T$. 

A meromorphic Higgs field $\Phi(x) dx$ on the Riemann sphere $C=\bP^1$ will have the general form
\begin{equation}\label{he}
	\Phi(x)=-\sum_{j=0}^{d_\infty-1}\Phi_{j+2}^\infty x^j+\sum_{k=1}^n\sum_{j=1}^{d_k+1}\frac{\Phi_j^k}{(x-x_k)^j}.
\end{equation}
Let $\pm y(x)$ be the eigenvalues of $\Phi(x)$, and $\pm\lambda=\pm y(x)\dd x$ the eigenvalues of $\Phi(x)\dd x$. Away from the poles of $\Phi(x)$ the  spectral curve $\Sigma$  is given by the condition $\det(\Phi(x)\dd x-\lambda\,\mathrm{id})=0$, i.e
\begin{equation}\label{eq:Sigma}
	y^2=Q_0(x),\qquad Q_0(x):=-\det\Phi(x)=\frac{1}{2}\tr\Phi(x)^2 .
\end{equation}

On the space of Higgs fields, there is a natural R-matrix Lie-Poisson bracket, given by the split rational classical R-matrix \cite{Reyman1979,Reyman1988} (see \cite{Bertola2022b} for more details in this context).
The so-called algebro-geometric Darboux coordinates $(q_i,p_i)$, where $i=1,\dots,g(\Sigma)$, provide a convenient system of coordinates to explicitly parametrise Higgs fields. For our current $\mathfrak{sl}_2$ case, these are defined by
\begin{equation}
	\Phi_{12}(q_i)=0,\qquad \Phi_{11}(q_i)=p_i.
\end{equation}
We see from \eqref{eq:Sigma} that $(q_i,p_i)$ are points on the spectral curve $\Sigma$, namely
\begin{equation}\label{eq:pq}
	p_i^2=Q_0(q_i).
\end{equation}

We introduce a reference connection through additional parameters $r_i$, obtained by shifting $p_i\mapsto p_i+\epsilon r_i$ in the Higgs field, analogously to the Painlev\'e I example in Section \ref{sec:pain}. In this way, the inhomogeneous term $g^{-1}dg$ in the gauge transformation \eqref{gauge} leads to a double pole of $Q(x) dx^{\tensor 2}$ at the point $x=q_i$ with residue $-p_i-\e r_i$, and the algebro-geometric Darboux coordinates  coincide with the variables $q_i,\,p_i$ of Section \ref{sec:opers}. If we have local monodromy exponents $\mu_k$, defined by
\begin{equation}
	\mu_k:=\Res_{x=x_k} y(x)\dd x,
\end{equation}
we will also include in the reference connection parameters $s_k$ defined by shifts $\mu_k\mapsto \mu_k+\left( s_k+\frac{1}{2}\right)$.
This shifting procedure gives a particular choice of reference connection $A_0$, and  immediately gives trivial isomonodromic flows of the form
\begin{equation}
	\widetilde{w}_i=\frac{\partial}{\partial p_i}-\frac{1}{\epsilon}\frac{\partial}{\partial r_i}.
\end{equation}

To have a complete set of coordinates, we also need to include a complete set of spectral invariants for the Higgs field $\Phi$. This problem was solved  in \cite{Bertola2022b} for a general nonresonant system, meaning that the matrix describing the leading behaviour of the Higgs field at a pole is diagonalizable with distinct eigenvalues. This  nonresonancy condition is generic but is not satisfied in the  examples corresponding to Painlev\'e I and III$_3$, so that in Section \ref{sec:three} we will use a slightly different set of coordinates.

We will denote the spectral invariants as $t_{j,k}$ and $H_{j,k}$, with the $t_{j,k}$ being Casimir elements, with $k=1,\dots,n$ indexing the poles $x_k$, and $j=1,\dots d_k$. The spectral invariants are given by
\begin{equation}\label{eq:tH}
	t_{j,k}:=-\Res (x-x_k)^jy(x)\dd x,\quad H_{j,k}:=-\Res_{x=x_k}\frac{y(x)\dd x}{j(x-x_k)^j},
\end{equation}
\begin{equation}
	t_{0,k}:=x_k,\qquad H_{0,k}:=\frac{1}{2} \Res_{x=x_k} \tr \Phi(x)^2 dx.
\end{equation}
For the case of a singular point that has been fixed at infinity,
\begin{equation}
	t_{j,\infty}:=-\Res x^{-j}y(x)\dd x,\qquad H_{j,\infty}:=-\frac{1}{j}\Res_{x=\infty} x^{j} y(x)\dd x.
\end{equation}

In the usual setting of isomonodromic deformations without a reference connection, $t_{j,k}$ are isomonodromic times, whose flows are generated
by the Hamiltonians $H_{j,k}$ through the R-matrix Poisson bracket. In our present context, these are all isomonodromic times: choosing as coordinates the $r$'s, $q$'s, $H$'s and $t$'s, and regarding $p_i$ as functions of the other variables through \eqref{eq:pq}, the flows $\widetilde{w}_i$ can be written as
\begin{equation}
	\widetilde{w}_i=\sum_{k=1}^n\sum_{j=1}^{d_k}\frac{\partial H_{j,k}}{\partial p_i}\bigg|_{q_i}\, \frac{\partial}{\partial H_{j,k}}-\frac{1}{\epsilon}\frac{\partial}{\partial r_i},
\end{equation}
where the subscript denotes partial derivatives with the variables $q_i$ fixed. By rescaling and taking linear combinations of the $\widetilde{w}_i$'s, the flows can be brought to the form $w_{j,k}=\frac{\partial}{\partial H_{j,k}}+\dots$. The other isomonodromic deformation equations can be obtained from the compatibility of the linear systems
\begin{equation}\label{eq:IsoDefHiggs}
\begin{cases}
\partial_x \Psi(x)=A(x)\Psi(x), \\
\partial_{t_{j,k}} \Psi(x)=B_{j,k}(x)\Psi(x),\qquad j=1,\dots,d_k,\,k=1,\dots,n,
\end{cases}
\end{equation}
and are given by
\begin{equation}\label{eq:LaxeqGen}
	\partial_{t_{l,k}}A-\partial_x B_{l,k}+[A,B_{l,k}]=0,\qquad \partial_{t_{l,k}}B_{j,l}-\partial_{t_{j,k}}B_{l,k}+[B_{j,l},B_{l,k}]=0,
\end{equation}
with the matrices $B_{j,k}$ themselves being fixed by the consistency conditions.

In fact,  the isomonodromic equations can be obtained in a uniform manner as the compatibility of the system
\begin{equation}
\begin{cases}
\partial_x \Psi(x)=A(x)\Psi(x), \\
\partial_{t_{j,k}} \Psi(x)=B_{j,k}(x)\Psi(x),& j=1,\dots,d_k,\,k=1,\dots,n,  \\
\partial_{r_i}\Psi(x)=0, & i=1,\dots, g(\Sigma).
\end{cases}
\end{equation}


\section{First example: Painlev\'e III$_3$}
\label{sec:three}

In this section we will study in detail the Joyce structure of class $S[A_1]$ labelled by the genus $g=0$ and pole orders $m=(3,3)$. This is the Joyce structure associated to the DT theory of the   Kr\"onecker 2-quiver.  The base parameterises quadratic differentials on $\bP^1$  with 2 poles of order 3, and simple zeros.
The corresponding Painlev\'e equation is the  III\,(D$_8$) equation, also known as Painlev\'e III$_3$, associated to perhaps the most famous class $S$ theory: pure $SU(2)$ super Yang-Mills, the original setting of Seiberg-Witten theory \cite{Seiberg1994,Bonelli2017}.

\subsection{Higgs field and quadratic differential}

Take $q,t\in \bC^*$ and $p\in \bC$ and consider the meromorphic Higgs field $\Phi=\Phi(x) \, \dd x$ on the trivial rank 2 vector bundle over $\bP^1$ defined by

\begin{equation}
\Phi(x)=\frac{1}{x^2}\left(\begin{array}{cc}
		0 & -q \\
		0 & 0
	\end{array} \right)+\frac{1}{x}\left(\begin{array}{cc}
		pq & 1 \\
		-q^{-1} & -pq
	\end{array} \right)+
	\left(\begin{array}{cc}
		0 & 0 \\
		t & 0
	\end{array} \right).
\end{equation}
Note  that $\Phi$ has double poles at $x=0$  and $x=\infty$ with nilpotent leading order behaviour. Thus $x=0$ and $x=\infty$ are irregular singularities of the connection $d-\Phi$ of Poincar\'e rank $\frac{1}{2}$.

Then matrix $\Phi(x)$ is a Lax matrix for the Painlev\'e III$_3$ equation (see for example \cite[Section 2]{Gavrylenko2017})
\begin{equation}\label{eq:PIII}
\frac{\dd^2q}{\dd t^2}=\frac{1}{q}\left(\frac{\dd q}{\dd t} \right)^2-\frac{1}{t}\frac{\dd q}{\dd t}+\frac{2q^2}{t}-\frac{2}{t^2},
\end{equation}
in the sense that \eqref{eq:PIII} describes isomonodromic deformation of the connection $\dd-\Phi$.

A slight complication in this case is that the poles are branch points of the spectral curve, so that one would need to modify the connection along the lines of \cite[Section 2.1]{Gavrylenko2017}  in order for equations \eqref{eq:tH} to hold. We will circumvent this issue by introducing the parameter $H$ simply as
\begin{equation}\label{eq:PIIIHam}
    H=\Res_{x=0} \left(\frac{x}{2} \tr\Phi(x)^2\,dx\right)=p^2q^2-tq-q^{-1}.
\end{equation}

Since $t\in \bC^*$, the spectral invariant
\begin{equation}\label{qu} \frac{1}{2} \tr(\Phi^2)= Q_0(x)  \, dx^{\tensor 2}= \frac{1}{x^2}\left(tx+H +x^{-1}\right) dx^{\tensor 2}\end{equation}
is a meromorphic quadratic differential  on $\bP^1$ with poles of order 3 at the points $x=0$ and $x=\infty$. Conversely any quadratic differential  on $\bP^1$ with  2 poles of order 3 and no other poles  can be put in the form \eqref{qu} by applying an automorphism of $\bP^1$.

The base $M$ of our Joyce structure will parameterise  differentials of the form \eqref{qu} which have simple zeros. Thus 
\begin{equation}
    M=\big\{(t,H)\in \bC^2: t(H^2-4t)\neq 0\big\}.\end{equation}
   One exceptional feature in this example is that every differential of the form \eqref{qu} is preserved by a unique non-trivial  automorphism of $\bP^1$, namely $x\mapsto (tx)^{-1}$. 
   In more  abstract language, the stack $\Quad(g,m)$ of quadratic differentials  has generic stabilizer $\bZ_2$ in this example \cite[Section 2.7]{BS}.

    We will construct a Joyce structure on $M$ by considering isomonodromic flows for meromorphic connections of the form 
\begin{equation}
\label{family}
	\nabla_\e= \dd-A_\e (x)\dd x, \qquad A_\e (x)=A_\infty(x)+\frac{1}{\e}\, \Phi(x), 
\end{equation}
where   $A_\epsilon(x)$  is  obtained from $\Phi(x)$ by the shift $p\mapsto p+\e rq^{-1}$. Explicitly
\begin{equation}
    \label{con2}
\nabla_\epsilon=\dd-  \left(\begin{array}{cc}
		r & 0 \\
		0 & -r
	\end{array} \right)\frac{\dd x}{x}-\frac{1}{\e}\left(\begin{array}{cc}
		pq & 1-qx^{-1} \\
		 tx-q^{-1}  & -pq
	\end{array} \right)\frac{\dd x}{x}. \end{equation}We view this  connection as  being parameterised by $\epsilon \in \bC^*$ together with the variables $(t,H,q,r)$, with $p$ defined implicitly by \eqref{eq:PIIIHam}.

\subsection{Extended isomonodromic flows}

Let us fix $\epsilon\in \bC^*$. If we move the point $(t,H)\in M$ we can vary $q=q(t,H)$ and $r=r(t,H)$ in such a way that the generalised monodromy of the connection \eqref{con2} is constant. The following result gives an explicit basis for the resulting two-dimensional space of  isomonodromic flows.

\begin{prop}
    The generalised monodromy of \eqref{con2} is preserved by the  vector fields 
\begin{equation}
\label{oneblue}
    w_1=t\frac{\partial}{\partial t}+\frac{2pq^2}{\e }\frac{\partial}{\partial q}+2qr\frac{\partial}{\partial q}-\frac{1}{2\e pq^2}(2rq^2 t-2r+q^2t)\frac{\partial}{\partial r},
\end{equation}
\begin{equation}\label{two1}
	w_2=\frac{\partial}{\partial H}-\frac{1}{2\e p q}\frac{\partial}{\partial r}.
\end{equation}
\end{prop}

\begin{proof}
Let us first consider the connection $d-A_\e(x) dx$ as being a function of  the variables $(t,q,p,r)$,  with $H$ determined by \eqref{eq:PIIIHam}.  We can describe isomonodromic deformations  of  this connection  with parameter $t$ as the compatibility conditions of the linear systems
\begin{equation}
\label{systems}
     \partial_x Y(x)=A_\e(x)Y(x), \qquad
 	\partial_t Y(x)=B_\e(x)Y(x). 
\end{equation}
The consistency condition
\begin{equation}
\label{consistency}
    \partial_t A_\e(x)-\partial_xB_\epsilon(x)+[A_\e(x),B_\e(x)]=0
\end{equation}
requires that we take
\begin{equation}
    B_\epsilon(x) =\frac{1}{\e t}\left(\begin{array}{cc}
	pq+\e r & 1 \\
	0 & -pq-\e r
\end{array} \right)+ \frac{x}{\e}\left(\begin{array}{cc}
	0 & 0 \\
	1 & 0
\end{array} \right),
\end{equation}
and gives the equations
\begin{equation}
\label{peg}
    t\frac{\dd q}{\dd t}=\frac{2pq^2}{\e}+2qr, \qquad 
    t\frac{\dd }{\dd t}(pq+\e r)=\frac{1}{\e}\left(q t-q^{-1} \right).
\end{equation}
Differentiating \eqref{eq:PIIIHam} then gives
\begin{equation}
    \frac{dH}{dt}=\frac{2r-q^2t(2r+1)}{tq}-2pq\epsilon\,  \frac{dr}{dt}.
\end{equation}

The equations \eqref{peg} do not uniquely  determine an isomonodromic flow because everything depends only on the combination $pq+\e r$, meaning that both $\partial_x-A_\e(x)$ and $\partial_t-B_\epsilon(x) $, and in particular the monodromy of $\nabla_\e $, are constant along the vector field
\begin{equation}
\label{w2}
    \widetilde{w_2}=\frac{\partial}{\partial r}-\frac{\e}{q}\frac{\partial}{\partial p}
\end{equation}
To fix this ambiguity  we impose that $\frac{\dd H}{\dd t}=0$ by adding a multiple of the vector field $\tilde{w}_2$.

We now rewrite the flows in the variables $(t,H,q,r)$, with $p$  determined by \eqref{eq:PIIIHam}. The first flow becomes
\begin{equation}
    \label{eq:isomonodromyPIII}
     t\frac{\dd q}{\dd t}=\frac{2pq^2}{\e}+2qr, \qquad  t\frac{\dd r}{\dd t}=\frac{2 r-q^2 t (2 r +1)}{2 p q^2  \e },
\end{equation}
which is generated by the vector field $w_1$. Moreover \eqref{w2} becomes
\begin{equation}
    \widetilde{w_2}=\frac{\partial}{\partial r}-2\e pq
    \frac{\partial}{\partial H},
\end{equation}
which is equivalent to $w_2$ by a simple rescaling.
\end{proof}
\begin{remark}
The combination $w_1-qtw_2$ gives the flow
\begin{equation}\label{eq:ReferencePIII}
    t\frac{dq}{dt}=\frac{2pq^2}{\e}+2qr,\qquad t\frac{dp}{dt}=-\frac{2p^2q-t+q^{-2}}{\e}+\frac{r \left(-4 p^2 q^3+2 q^2 t-2\right)}{2 p q^3},\qquad t\frac{dH}{dt}=-tq.
\end{equation}
When $r=0$ this reproduces the usual Painlev\'e flow defined by the Hamiltonian $H$ and the symplectic form $tdq\wedge dp$. 
In contrast to what happens in the usual treatment of Painlev\'e equations, however, in the presence of a reference connection the flow in $t$ can be normalised so that $H$ is conserved. However, since $q$ is unaffected by the flow $w_2$, irrespective of how we normalise the flow, the second order isomonodromy equation for $q$ is still the ($\e$-deformed) Painlev\'e III$_3$ equation 
\begin{equation}\label{pa}
	\frac{\dd^2 q}{\dd t^2}=\frac{1}{q}\left(\frac{\dd q}{\dd t} \right)^2-\frac{1}{t}\frac{\dd q}{\dd t}+\frac{1}{\e^2}\left(\frac{2q^2}{t}-\frac{2}{t^2}\right)
\end{equation}
This phenomenon was observed in the Painlev\'e I example in \cite[Section 6.1]{DM2}.
\end{remark}

\subsection{Spectral curve and its homology}

Let us fix a point $(t,H)\in M$. The spectral curve $\Sigma$  is a smooth curve of genus one which is a double cover $p\colon \Sigma\to \bP^1$  branched over 4 points of $\bP^1$. We denote by $\sigma\colon \Sigma\to \Sigma$ the covering involution. The branch points of $p$ include the points $0,\infty\in \bP^1$, and we will also write $0,\infty\in \Sigma$ for the unique points of $\Sigma$ lying over them. 

The open subset $\Sigma^0\subset \Sigma$ is  the complement $\Sigma^0=\Sigma\setminus\{0,\infty\}$. This subset  can be identified with the set of points $(x,y)\in \bC^*\times \bC$ satisfying
\begin{equation}\label{eq:SigmaPIII}
	y^2=Q_0(x)=\frac{t}{x}+\frac{H}{x^2}+\frac{1}{x^3}.
\end{equation}
The projection is then $p(x,y)=x$ and the covering involution is $\sigma(x,y)=(x,-y)$.

\begin{lemma}
\label{nak}
    The inclusion $i\colon \Sigma^0\hookrightarrow \Sigma$ induces an embedding
of anti-invariant homology groups \begin{equation}\label{istar}i_*\colon H_1(\Sigma^0,\bZ)^-\hookrightarrow H_1(\Sigma,\bZ)^-=H_1(\Sigma,\bZ)\end{equation}
whose image  has index 2. 
\end{lemma}

\begin{proof}
By uniformisation we can identify $\Sigma$ with the quotient of $\bC$ by a lattice generated by $\omega_1,\omega_2\in \bC^*$ with $\Im(\omega_2/\omega_1)>0$. We can assume that $0\in \bC$  corresponds to $0\in \Sigma$, and that the  involution $\sigma$ is induced by $z\mapsto -z$ on $\bC$. Then $\infty\in \Sigma$ will correspond to some half-period, without loss of generality $\half(\omega_1+\omega_2)$. The homology group $H_1(\Sigma,\bZ)$ has a basis consisting of cycles $\alpha_1,\alpha_2$ which we can represent by the images of the closed intervals $[0,\omega_i]\subset \bC$. It is then obvious that $\sigma^*(\alpha_i)=-\alpha_i$ and hence that $H_1(\Sigma,\bZ)^-=H_1(\Sigma,\bZ)$.

 The inclusion $i\colon \Sigma^0\hookrightarrow \Sigma$ induces a surjection $i_*\colon H_1(\Sigma^0,\bZ)\to H_1(\Sigma,\bZ)$ whose kernel is freely generated by the class $\delta$ defined by a small anti-clockwise loop  around the point $\half(\omega_1+\omega_2)$.
Let us lift the cycles $\alpha_i\in H_1(\Sigma,\bZ)$ to elements $\tilde{\alpha}_i \in H_1(\Sigma^0,\bZ)$. Since the $\alpha_i$ pass through $0\in \bC$  this requires us to first perturb them in some way. There is no canonical way to do this but we can take for example $\tilde{\alpha}_1=\tfrac{1}{4}\omega_2+[0,\omega_1]$ and $\tilde{\alpha}_2=\tfrac{1}{4}\omega_1+[0,\omega_2]$. The group $H_1(\Sigma^0,\bZ)$ then has a basis consisting of $\tilde{\alpha}_1,\tilde{\alpha}_2$ and $\delta$.

It is easy to see that $\sigma^*(\tilde{\alpha}_1)+\tilde{\alpha}_1=\delta$ and $\sigma^*(\tilde{\alpha}_2)+\tilde{\alpha}_2=-\delta$. Also $\sigma^*(\delta)=\delta$. 
Thus the anti-invariant subgroup $H_1(\Sigma^0,\bZ)^-\subset H_1(\Sigma^0,\bZ)$ is spanned by $\tilde{\alpha}_1+\tilde{\alpha}_2$ and $\tilde{\alpha}_1-\tilde{\alpha}_2-\delta$. It follows that the map \eqref{istar} is injective, and its image is the index 2  subgroup  spanned by $\alpha_1+\alpha_2$ and $\alpha_1-\alpha_2$. 
\end{proof}

In what follows we will
 fix a basis $(\gamma_1,\gamma_2)$ for the group $H_1(\Sigma^0,\bZ)^-$. The intersection form on $H_1(\Sigma,\bZ)$ pulls back along $i_*$ to give a skew-symmetric pairing  on  $H_1(\Sigma^0,\bZ)^-$.  Since the image of $i_*$ has   index 2 we can order our basis so that $\<\gamma_1,\gamma_2\>=2$. 

\subsection{Forms and periods}
By the gauge transformation \eqref{gauge} we can rewrite \eqref{con2} as a scalar equation of the form \eqref{above2}. In this case we have $Q_0(x)$ as in \eqref{qu} above, and
\begin{equation}\label{nf}
 Q_1(x)=-\frac{pq^2}{x^2(x-q)}+\frac{2pqr}{x^2}, \qquad Q_2(x)=\frac{3}{4(x-q)^2}-\frac{x+rq}{x^2(x-q)}+\frac{r^2}{x^2}.
\end{equation}
As in Section \ref{sec:Quad} we introduce  meromorphic differentials on $\Sigma$ by the expressions
\begin{equation}\label{nff}\lambda = y\, \dd x, \qquad \theta=-\frac{Q_1(x) \, dx }{2\sqrt{Q_0(x)} }.\end{equation}

The Seiberg-Witten differential $\lambda$ has double poles at the points $0,\infty\in \Sigma$, and
no other poles. 
The differential $\theta$  has simple poles at the points $(q,\pm p)$ with residues $\pm \half$ respectively, and no other poles. 
We introduce the periods 
\begin{equation}\label{pers}z_i=\oint_{\gamma_i} \lambda, \qquad \theta_i=\oint_{\gamma_i} \theta.\end{equation}
Note that the expressions $\theta_i$ are only well-defined up to the addition of integer multiples of $2\pi i$. To specify them we must lift the  cycles $\gamma_i$ to the complement $\Sigma^*=\Sigma^0\setminus \{(q,\pm p)\}$. We can always do this so that the resulting elements $\gamma^*_i\in H_1(\Sigma^*,\bZ)$ are anti-invariant. Since $\theta$ has residues $\pm \half$ at the points $(q,\pm p)$, taking different  lifts changes the result by a multiple of $2\pi i$. 

Let us set $s=\log(t)$ and take $(s,H)$ as local co-ordinates  on the base $M$. There are 
corresponding linear coordinates $(\theta_s,\,\theta_H)$ on the fibres of the projection $\pi\colon X=T_M\to M$ defined by the relation\begin{equation}\label{fafa}
\theta_s\, \frac{\partial }{\partial s} +\theta_H\, \frac{\partial }{\partial H} =\theta_1\,  \frac{\partial }{\partial z_1} +\theta_2\, \frac{\partial }{\partial z_2} .
\end{equation}
 
 Introduce the differentials
\begin{equation}\label{fafafa}\omega=\frac{\partial \lambda}{\partial H} =\frac{dx}{2x^2y}, \qquad \beta=\frac{\partial \lambda}{\partial s} =\frac{t\, dx}{2x y}.\end{equation}
Then $\omega$ is a holomorphic differential on $\Sigma$, and $\beta$ has a pole of order 2  with vanishing residue at the branch-point $x = \infty$  and no other poles. We denote the corresponding periods by $\omega_{i}$ and $\beta_i$ respectively. There is an obvious relation
\begin{equation}\label{fink}dz_i=\frac{\partial z_i}{\partial s} \, ds+\frac{\partial z_i}{\partial H} \, dH=\beta_i \, ds+\omega_i \, dH,\end{equation}
and using \eqref{fafa} we get
\begin{equation}\label{eq:thetatHdef}\theta_i=\beta_i\theta_s+\omega_i\theta_H.\end{equation}

\subsection{Computation of periods}In this section we use the Riemann bilinear relations to give explicit expressions for $\theta_s$ and $\theta_H$.
  
\begin{lemma}\label{thm:lemmathetaIII}
   There are formulae
    \begin{equation}\label{ths}\theta_s=\frac{1}{2} \int_{(q,-p)}^{(q,p)} \frac{dx}{2x^2y}, \qquad  \theta_H= -\frac{1}{2} \int_{(q,-p)}^{(q,p)} \frac{t\, dx}{2xy}-2pqr.\end{equation}
\end{lemma}

\begin{proof}
Note that $y=\sqrt{Q_0(x)}$ is a uniformising parameter at the point $\infty\in \Sigma$. We denote by  $x^{1/2}$  the  choice of square-root near this point for which  $y\simeq t^{1/2}x^{-1/2}$. The differentials introduced above have leading order behaviour
\begin{equation}\label{lord}
    \omega=\frac{\dd x}{2t^{1/2}x^{3/2}}+O\left(\frac{\dd x}{x^{5/2}} \right), \qquad \beta=\frac{t^{1/2} \dd x}{2x^{1/2}}+O\left(\frac{\dd x}{x^{3/2}} \right), \end{equation}
    \begin{equation}\qquad \theta=-\frac{pqr \, dx}{t^{1/2}x^{3/2}}+O\left(\frac{\dd x}{x^{5/2}} \right).
\end{equation}

The Riemann bilinear relations for meromorphic differentials (see Proposition 3.4.1 in \cite{BertolaLectures}, a direct consequence e.g. of equations 3.0.1 and 3.0.2 in \cite{FarkasKra})   give
 \begin{equation}\label{l1a}\omega_1\beta_{2}-\omega_2\beta_{1} = 2\pi i \cdot \langle\gamma_1,\gamma_2\rangle\cdot  \Res_{x=\infty}\left( \beta\int^x\omega\right). \end{equation}
Note that although this formula is usually given in the case that the cycles $\gamma_1,\gamma_2\in H_1(\Sigma,\bZ)$ satisfy $\<\gamma_1,\gamma_2\>=1$, the case of an arbitrary pair of linearly independent cycles follows immediately from this. Note also that the form $dx/x$ has residue $-2$ at the point $\infty\in \Sigma$.
Using  $\<\gamma_1,\gamma_2\>=2$ and the expressions \eqref{lord} we then get
\begin{equation}\label{l1}\langle\omega,\beta\rangle:=\omega_1\beta_{2}-\omega_2\beta_{1} =4\pi i.\end{equation}

We can write similar expressions
 \begin{equation}\label{lorda}\langle\omega,\theta\rangle  :=\omega_1\theta_2-\omega_2\theta_1 = 4\pi i \sum_{x_i\in \{(q,\pm p)\}}\Res_{x=x_i}\left(\theta\int^x\omega\right)
 =2\pi i  \int_{(q,-p)}^{(q,p)}\omega,\end{equation}
 \begin{equation}\label{lordb} \langle\beta,\theta\rangle:=\beta_{1}\theta_2-\beta_{2}\theta_1 = 4\pi i \sum_{x_i\in\{(q,\pm p),\infty\}}\Res_{x=x_i}\left(\theta\int^x\beta \right)=2\pi i \int_{(q,-p)}^{(q,p)}\beta+8\pi i pqr.
\end{equation}
Using \eqref{eq:thetatHdef} we then have
\begin{equation}
   2\pi i  \int_{(q,-p)}^{(q,p)}\omega=  \langle\omega,\theta\rangle=\theta_s\langle\omega,\beta\rangle=4\pi i\theta_s,
\end{equation}
\begin{equation}
     2\pi i\int_{(q,-p)}^{(q,p)}\beta+8\pi i pqr=\langle\beta,\theta\rangle=\theta_H\langle\beta,\omega\rangle=-4\pi i\theta_H,
\end{equation}
and comparing with \eqref{fafafa} gives the result.
\end{proof}

\begin{remark}
    Since the $\theta_i$ are only well-defined up to the addition of integer multiples of  $2\pi i$, it follows from \eqref{eq:thetatHdef} and \eqref{l1} that the pair $(\theta_{s},\theta_H)$ is only well-defined up to the addition of integer combinations  of the vectors $(\omega_i,-\beta_i)$. This matches the indeterminacy in the choice of path in the formulae \eqref{ths} which should be taken to be anti-invariant under the covering involution $\sigma$.
\end{remark}
   
Recall from \eqref{poisson} that the Poisson structure on $M$ is defined by the condition \begin{equation}\{z_1,z_2\}=2\pi i \, \<\gamma_1,\gamma_2\>=4\pi i.\end{equation} Inverting this gives the symplectic form  
\begin{equation}\label{handy} \Omega_0=-\frac{1}{4\pi i} \cdot dz_1\wedge dz_2 = \frac{1}{4\pi i}  (\omega_1\beta_{2}-\omega_2\beta_{1})\cdot ds\wedge dH=ds\wedge dH,\end{equation}
where we used  \eqref{fink} and \eqref{l1}.

\subsection{Abelian holonomy map}

Denote by $\Xi$ the change of coordinates $(s,H,q,r)\mapsto(s,H,\theta_s,\theta_H)$.
Applying the chain rule  to Lemma \ref{thm:lemmathetaIII} gives
	\begin{equation}\label{eq:dstodthetaIII}
	\Xi_*\left(\frac{\partial}{\partial s} \right)=\frac{\partial}{\partial s}-t\kappa_3 \frac{\partial}{\partial\theta_s}+\left(
 -\frac{t}{4} \int_{(q,-p)}^{(q,p)} \frac{dx}{xy}+t^2\kappa_2 -\frac{rt}{p} \right)\frac{\partial}{\partial\theta_H},
\end{equation}
\begin{equation}\label{eq:dqtodthetaIII}
\Xi_* \Big(\frac{\partial}{\partial H}\Big)=\frac{\partial}{\partial H}-\kappa_4 \frac{\partial}{\partial \theta_s} + \Big(t \kappa_3 -\frac{r}{qp}\Big)\frac{\partial}{\partial \theta_H},\end{equation}
\begin{equation}\Xi_* \Big(\frac{\partial}{\partial q}\Big)=\frac{1}{2q^2p} \frac{\partial}{\partial \theta_s} +\frac{2r-2trq^2-tq^2}{2pq^3}\frac{\partial}{\partial \theta_H},\qquad \label{eq:drtodthetaIII}
\Xi_* \Big(\frac{\partial}{\partial r}\Big)=-2qp \frac{\partial}{\partial \theta_H},\end{equation}
where we introduced
\begin{equation}\kappa_i:=\frac{1}{8}\int_{(q,-p)}^{(q,p)} \frac{dx}{x^i y^3},\qquad i=2,3,4.\end{equation}
Inverting the relations \eqref{eq:dqtodthetaIII} and \eqref{eq:drtodthetaIII} gives
\begin{equation}
\label{two2}\Xi_*^{-1} \Big(\frac{\partial}{\partial \theta_s}\Big)={2q^2p} \frac{\partial}{\partial q} +\frac{2r-2trq^2-tq^2}{2pq^2}\frac{\partial}{\partial r},\qquad 
\Xi_*^{-1} \Big(\frac{\partial}{\partial \theta_H}\Big)=-\frac{1}{2qp}\frac{\partial}{\partial r}.\end{equation}

Fix a point $(t,H)\in M$ and a basis $(\gamma_1,\gamma_2)\in H^1(\Sigma^0,\bZ)^-$ as above.

\begin{lemma}\label{abhol}
      The  map
    \begin{equation}\label{map} \big\{(q,p)\in (\bC^*)^2: p^2=Q_0(q)\big\}\times \big\{r\in \bC\big\}\to \big\{(e^{\theta_1},e^{\theta_2})\big\}\in (\bC^*)^2\end{equation}
    defined by the equations \eqref{nf}, \eqref{nff} and \eqref{pers} is  an unramified double cover.
\end{lemma}

\begin{proof}  The degree 0 divisor $(q,p)-\infty$ on $\Sigma$ defines a line bundle whose local sections are meromorphic functions on $\Sigma$ with a zero at $(q,p)$ and at worst a simple pole at  $\infty$.
    Consider the form \begin{equation}
    \label{thnew}
       \theta'=\theta+\frac{dx}{2(x-q)}
    \end{equation} on $\Sigma$. It has a simple pole at $(q,p)$ with residue $+1$, a simple pole  at  $\infty$ with residue $-1$, and no other poles. It follows that $\partial=d-\theta'$ defines a holomorphic connection on $L$. Since  $\theta'-\theta$  is pulled back from $\bP^1$, the form $\theta'$ has the same periods as $\theta$. Thus the monodromy of the connection $(L,\partial)$   around the cycle $\gamma_i\in H_1(\Sigma^0,\bZ)^-$ is given by multiplication by $\exp(\theta_i)\in \bC^*$.

    Suppose there are three points $(q_i,p_i,r_i)$ with  $i=1,2,3$ in the source of \eqref{map} which are mapped to the same point. Since the cycles $\gamma_1,\gamma_2$ span an index 2  subgroup of $H_1(\Sigma,\bZ)$, at least two of the points (without loss of generality corresponding to $i=1,2$) must correspond to line bundles with connection $(L_i,\partial_i)$ which have identical monodromy around all cycles in $H_1(\Sigma,\bZ)$. The abelian Riemann-Hilbert correspondence tells us that such line bundles with connection $(L_i,\nabla_i)$ are isomorphic. But it is a well-known fact (easily proved using the Riemann-Roch theorem) that the line bundles $L_i$ defined by different points $(q_i,p_i)\in \Sigma$ are all non-isomorphic. So we conclude that $(q_1,p_1)=(q_2,p_2)$. But the formulae \eqref{eq:thetatHdef} - \eqref{ths} then imply that also $r_1=r_2$. Indeed $\omega_1,\omega_2\in \bC^*$ do not lie on the same ray, so changing $r$ cannot move both $\theta_i$ by an imaginary number.

  The derivative of the map \eqref{map} is computed by \eqref{eq:dstodthetaIII} and \eqref{eq:dqtodthetaIII} together with \eqref{eq:thetatHdef} and is everywhere invertible.  Thus \eqref{map} is a finite unramified cover of degree at most 2. Now it is easy to check that our map is invariant under the  involution
    \begin{equation}\label{inv}
         (q,p,r)\mapsto \left((qt)^{-1},-tpq^2 ,-(r+\half)\right),
    \end{equation}
    and this involution is moreover fixed-point free. Therefore the map is a double cover.
\end{proof}

\subsection{Isomonodromic flows in new coordinates}

We now write the isomonodromic flows in coordinates $(s,H,\theta_s,\theta_H)$.

\begin{proposition}
\label{propIII}
  The push-forward of the isomonodromic flows along $\Xi$ can be written in the form
  \begin{equation}\label{11}
    \Xi_*(w_1)=\frac{\partial}{\partial s}+\frac{1}{\e}\frac{\partial}{\partial\theta_s}+\frac{\partial^2K}{\partial\theta_s\partial\theta_H}\frac{\partial}{\partial\theta_s}-\frac{\partial^2K}{\partial\theta_s^2}\frac{\partial}{\partial\theta_H},
\end{equation}
\begin{equation}\label{12}
    \Xi_*(w_2)=\frac{\partial}{\partial H}+\frac{1}{\e}\frac{\partial}{\partial\theta_H}+\frac{\partial^2K}{\partial\theta_H^2}\frac{\partial}{\partial\theta_s}-\frac{\partial^2K}{\partial\theta_s\partial\theta_H}\frac{\partial}{\partial\theta_H},
\end{equation}
where $K\colon X^\hash\to \bC$ is a single locally-defined function.
\end{proposition}

\begin{proof}
The formulae \eqref{oneblue} - \eqref{two1} for the isomonodromic flows and the relations \eqref{eq:dstodthetaIII} - \eqref{eq:drtodthetaIII} give

\begin{gather}
\Xi_*(w_1)=\frac{\partial}{\partial s}+\frac{1}{\e}\frac{\partial}{\partial \theta_s}+b\frac{\partial}{\partial \theta_s} - a\frac{\partial}{\partial \theta_H},\\
\Xi_*(w_2)=\frac{\partial}{\partial H}+\frac{1}{\e}\frac{\partial}{\partial \theta_H}+c\frac{\partial}{\partial \theta_s} -b\frac{\partial}{\partial \theta_H},\end{gather}
where we introduced
\begin{equation}
\label{jimmy}
     a=\frac{t}{4} \int_{(q,-p)}^{(q,p)} \frac{dx}{xy}-t^2\kappa_2+\frac{2rt}{p}+\frac{2r^2(tq^2-1)}{pq^2}, \qquad b=\frac{r}{qp}-t\kappa_3,\qquad 
c=-\kappa_4. \end{equation}

Using \eqref{two2} we find that 
\begin{equation}
\label{sno}
    \frac{\partial b}{\partial \theta_s}=\frac{\partial a}{\partial \theta_H}, \qquad \frac{\partial c}{\partial \theta_s}=\frac{\partial b}{\partial \theta_H}.
\end{equation}
This implies the existence of a single locally-defined  function $K$ such that
    \begin{equation}\label{park}
\frac{\partial^2 K}{\partial \theta_s^2}=a , \qquad \frac{\partial^2 K}{\partial \theta_s \partial \theta_H}=b,\qquad \frac{\partial^2 K}{\partial \theta_H^2}=c,\end{equation}
and the result follows. 
\end{proof}

Applying the operators \eqref{two2} to the expressions \eqref{jimmy} gives
\begin{equation}\label{one'}\frac{\partial^3 K}{\partial\theta_s^3} =-\frac{3 t^2}{2 p^2}+qt +\frac{2 r t}{p^2q^2} \left(3-3 tq^2+2 p^2q^3\right)+\frac{2 r^2 }{p^2q^4}\left(-1+2Hq+10tq^2+2Htq^3 -t^2 q^4 \right),
\end{equation}
\begin{equation}
\label{one}
\frac{\partial^3 K}{\partial \theta_s^2 \partial \theta_H}=-\frac{t}{p^2 q}+\frac{2r}{p^2q^3}\left(1- tq^2 \right), \qquad  \frac{\partial^3 K}{\partial \theta_s \partial \theta_H^2}=-\frac{1}{2q^2p^2}, \qquad \frac{\partial^3 K}{\partial \theta_H^3}=0.\end{equation}

\subsection{ Pleba{\'n}ski function}\label{sec:PIIIPleb}

While $K$ is the generating function for the flows $\Xi_*(w_i)$, it is not a simple function of our original variables $(p,q,t,r)$. In fact, integrating the above equations it is easy to see that $K$ will be expressed in terms of complicated combinations of elliptic functions. 

\begin{theorem}\label{thm:Pleba{'n}skiIII}
	In the co-ordinates $(z_1,z_2,\theta_1,\theta_2)$ the isomonodromy flows take the form
    \begin{equation}
\label{13}
	\frac{\partial}{\partial z_i}+\frac{1}{\e}\cdot \frac{\partial}{\partial\theta_i}+4\pi i\cdot \frac{\partial^2 W}{\partial\theta_i\partial\theta_1}\cdot \frac{\partial}{\partial\theta_2}- 4\pi i\cdot\frac{\partial^2 W}{\partial\theta_i\partial\theta_2}\cdot \frac{\partial}{\partial\theta_1},
\end{equation}
    where the Pleba{\'n}ski function $W$ is given by
    \begin{equation}\label{eq:Pleba{'n}skiPIII}
W=\frac{pq}{6(H^2-4t)}\big(tq+(H+6tq)r+(6H+12tq)r^2+8p^2q^2r^3\big).\end{equation}
\end{theorem}

\begin{proof}
Let us temporarily introduce variables $u_1=H\cdot \sqrt{4\pi i}$ and $u_2=s\cdot \sqrt{4\pi i}$. We also set $\phi_1=\theta_H\cdot \sqrt{4\pi i}$ and $\phi_2=\theta_s\cdot \sqrt{4\pi i}$. Then by \eqref{handy},  the change of variables $(z_1,z_2) \mapsto (u_1,u_2)$ is symplectic. Dividing  \eqref{11} - \eqref{12} by $\sqrt{4\pi i}$ and making the trivial change from $(s,H)\mapsto (u_1,u_2)$ shows that the isomonodromy flows are spanned by the vector fields
\begin{equation}\frac{\partial}{\partial u_i}+\frac{1}{\e}\cdot \frac{\partial}{\partial\phi_i}+4\pi i\cdot  \frac{\partial^2 K}{\partial\phi_i\partial\phi_1}\cdot \frac{\partial}{\partial\phi_2}- 4\pi i \cdot \frac{\partial^2 K}{\partial\phi_i\partial\phi_2}\cdot \frac{\partial}{\partial\phi_1}.\end{equation}
Applying  \cite[Prop. 4.2]{A2} to the variable change $(z_1,z_2) \mapsto (u_1,u_2)$ then gives the flows in the form \eqref{13}.
Note that the resulting function $W$ is only well-defined up to the addition of linear functions in the $\theta$ variables.

According to \cite[Prop. 4.2]{A2} the functions  $K$ and $W$ differ by a cubic function of the $\theta$ variables. In particular,  the fourth derivatives of $K$ and $W$ in the $\theta$ variables   necessarily coincide.   Now, the fourth derivatives of $K$ can be computed by applying the operators \eqref{two2} to \eqref{one'} - \eqref{one}, and it is then possible to check that these coincide with the fourth derivatives of the expression $\tilde{W}$ on the right-hand side of \eqref{eq:Pleba{'n}skiPIII}. Thus $\tilde{W}$  can differ from  $W$  by at most quartic functions of the $\theta$ variables.  But the second derivatives of $W$ in the $\theta$ variables  are necessarily  periodic in $\theta$. Since $\tilde{W}$ is also periodic in $\theta$  it must coincide with $W$ up to the addition of linear functions in the $\theta$ variables, which is the required statement, since $W$ was only well-defined up to this indeterminacy.
\end{proof}

We can directly check the properties \eqref{re2} - \eqref{per}  of the Pleba{\'n}ski function. From the change of variables  between  the co-ordinates $(z_1,z_2,\theta_1,\theta_2)$ and $(t,H,q,r)$, it is easy to see that  there is an equality of vector fields on $X^\hash$ 
\begin{equation}
    \label{euler}E:=z_1\frac{\partial}{\partial z_1}+ z_2\frac{\partial}{\partial z_1}=4\frac{\partial}{\partial s}+2H\frac{\partial}{\partial H}  -2q\frac{\partial}{\partial q}.\end{equation}
    Thus \eqref{st2} is the statement that $W$ is homogeneous of weight $-1$ under rescaling $(t,H,q,r)$ with weights $(4,2,-2,0)$. Note that by \eqref{eq:PIIIHam} this also involves rescaling the implicit variable $p$ with weight $3$.
Secondly, \eqref{ths} shows that the map $(q,p)\mapsto (q,-p)$ changes the sign of $(\theta_1,\theta_2)$. Thus \eqref{re2} corresponds to the fact that $W$ is an odd function of the variable $p$. Finally $W$ has the periodicity  property \eqref{per} because by Lemma \ref{abhol} it is a function of $(z_1,z_2,e^{\theta_1},e^{\theta_2})$.


\subsection{Pleba{\'n}ski function near $\theta=0$}

We would like to understand the behaviour of the potentials $W$ and $K$ near the zero section $M\subset T_M$ defined by $\theta_t=\theta_H=0$. This is not completely straightforward, because the potentials are expressed in the variables $(t,H,q,r)$ rather than $(t,H,\theta_t,\theta_H)$. To study this limit, we will first uniformise the spectral curve $\Sigma$ defined by \eqref{eq:SigmaPIII}. Observe that if we introduce
\begin{equation}
X=tx+\tfrac{1}{3}H,\qquad	Y=2tx^2y,
\end{equation}
then $\Sigma$   is brought into the standard Weierstrass form
\begin{equation}
	Y^2=4X^3-g_2X-g_3,\qquad g_2=\frac{1}{3}\left(4H^2-12t\right),\qquad g_3=\frac{4H}{27}\left(9t-2H^2 \right).
\end{equation}
It  can then be uniformised by writing
\begin{equation}
	X=\wp(u),\qquad Y=\wp'(u),
\end{equation}
with $\wp(u)$ the Weierstrass $\wp$-function.
The differentials $\omega,\,\beta$  are given by
\begin{equation}
	\omega=\frac{\dd x}{2x^2y}=\frac{\dd X}{Y}=\dd u,\qquad \beta=\frac{t \dd x}{2xy}=xt\omega=\frac{1}{3}\left(3\wp(u)-H\right)du.
\end{equation}

Denote by $v\in \bC$ a point  corresponding to the point $(q,p)\in\Sigma$, so that
\begin{equation}\label{eq:pqvIII}
	q=\frac{1}{3t}\left(3\wp(v)-H \right),\qquad p=\frac{9t\wp'(v)}{2\left(3\wp(v)-H \right)^2}.
\end{equation}
Then we have
\begin{equation}\label{eq:thetasuni}
	\theta_s=\frac{1}{2}\int_{(q,-p)}^{(q,p)}\omega=v,
\qquad 
	\theta_H=-\frac{1}{2}\int_{(q,-p)}^{(q,p)}\beta-2pqr=\zeta(v)+\frac{Hv}{3}-\frac{3r\wp'(v)}{3\wp(v)-H},
\end{equation}
where $\zeta(v)$ denotes the Weierstrass $\zeta$-function.

We recall from Section \ref{summy} that, assuming the Pleba{\'n}ski function $W$ is regular along the locus $\theta_s=\theta_H=0$, there is a locally-defined function $S=S(t,H)$ such that
\begin{equation}
    W=\left(\frac{\partial S}{\partial t}\right)\theta_t+\left(\frac{\partial S}{\partial H} \right)\theta_H+O(\theta^3)
\end{equation}
for small $\theta_t$ and $\theta_H$. Recall also the definition \eqref{st} of the linear Joyce connection $\nabla^J$.

\begin{theorem}\label{thm:Pleba{'n}skiOperIII}
\begin{itemize}
    \item[(i)]The Pleba{\'n}ski function $W$ is regular on the locus $\theta_s=\theta_H=0$.
    
    \item[(ii)] There is an identity
     \begin{equation}S= \log (H^2-4t)^{-\frac{1}{24}}.\end{equation}
     
 \item[(iii)] The co-ordinates $(s,H)$ on $M$ are flat for the linear Joyce connection.
 \end{itemize}
\end{theorem}


\begin{proof}
      Let us introduce the function  $w=(1+2r)/v$ and use $(v,w)$ as co-ordinates in place of $(q,r)$. For small $v$ and $w$  the equations \eqref{eq:thetasuni}  together with  the Laurent expansions of the Weierstrass functions give
\begin{equation}
	\theta_s=v, \qquad \theta_H=w+\frac{Hwv^2}{3}+O((v,w)^3).
\end{equation}
In particular, the locus  $v=w=0$ coincides with $\theta_s=\theta_H=0$. 
The equations \eqref{eq:pqvIII}  give
\begin{equation}
\label{pesto}
    q=\frac{1}{tv^2}-\frac{H}{3t}+\frac{(H^2-3t)v^2}{15t}+O(v^4), \qquad 
    pq=-\frac{1}{v}-\frac{Hv}{3}+O(v^3),\qquad r=-\frac{1}{2}+\frac{wv}{2}.
\end{equation} 
 We can now
substitute these formulae into the explicit expression \eqref{eq:Pleba{'n}skiPIII} to compute the leading behaviour of the Pleba\'nski function $W$. We find that
\begin{equation}
	 	W=\frac{1}{12(H^2-4t)} \left(2tv-Hw\right) +O((v,w)^3)
	\end{equation}
 which implies (i) and (ii). For (iii) we substitute the formulae \eqref{pesto}  into \eqref{one'} - \eqref{one} and find  that the third derivatives of $K$ in the $\theta$ directions  vanish along $\theta=0$. As explained in the proof of \cite[Theorem 4.5]{A2} this then implies that $(s,H)$ are flat for the linear  Joyce connection.
\end{proof}

\subsection{An alternative approach: poles of Painlev\'e III}

Our computation of the limit $\theta\rightarrow0$  relied heavily on the explicit uniformisation of the spectral curve by elliptic functions. It amounted to approaching the $\theta=0$ locus along purely vertical directions with fixed $t,H$. We can instead approach this locus along the two isomonodromic directions spanned by $w_1$ and $w_2$. Noting that $v\rightarrow 0$ amounts to a double pole of $q$, we can then identify the locus $v=0$ with a pole of the Painlev\'e equation \eqref{eq:PIII}. Let $t_0$ be a position of such a pole. The isomonodromic system \eqref{eq:isomonodromyPIII} implies the following behaviour:
\begin{equation}\label{eq:pqpoleIII}
    q(t)=\frac{t_0\e^2}{(t-t_0)^2}+q_0+O(t-t_0),\qquad p(t)=-\frac{(t-t_0)}{\e}-\frac{(t-t_0)^2}{2t_0\e}
    +O((t-t_0)^2),
\end{equation}
\begin{equation}
	r(t)=-\frac{1}{2}+\rho(t-t_0)+O(t-t_0)^2,\qquad H(t)=p^2q^2-tq-q^{-1}=H_0+2t_0 r_0 +O(t-t_0)^3
\end{equation}
where \begin{equation}\label{eq:HpoleIII}
  H_0=-3q_0t_0+\frac{\epsilon^2}{4}  ,
\end{equation} and the remaining parameters $(\rho,q_0)$ are  left undetermined by the system. We can then view $(t_0,H_0)$ as specifying a point of $M$ which is the projection of the relevant point of the $\theta=0$ locus.

Note that only the leading order behaviour coincides with the one obtained by uniformization, as the isomonodromic flows do not preserve the spectral curve. Nonetheless, one obtains the identifications, valid in the leading order in the expansion
\begin{equation}
	v\simeq \frac{t-t_0}{\e t_0},\qquad w\simeq 2\e t_0 \rho.
\end{equation}
As could be expected, the isomonodromic formula provides an $\e$-deformation of the one \eqref{eq:pqvIII} coming from uniformisation of the spectral curve. However we are simply approaching the same locus from a different direction, so unsurprisingly the end result  is the same, and this leads to an alternative derivation of Theorem \ref{thm:Pleba{'n}skiOperIII}.


\section{Second example: Painlev\'e II}
\label{sec:two}

Our second example will be the Joyce structure of class $S[A_1]$ corresponding to genus $g=0$ and pole order $m=(8)$. This Joyce structure is related to the DT theory of the A$_3$ quiver and to the  Painlev\'e II equation. The corresponding field theory of class $S$ is known as Argyres-Douglas theory $H_0$, or the $(A_1,A_3)$ theory \cite{Bonelli2017}. The relevant quadratic differentials take the form
\begin{equation}\label{q0}Q_0(x)  dx^{\tensor 2}=(x^4+tx^2-2 \alpha x+2H )dx^{\tensor 2}\end{equation}
and the associated flat connections feature both nontrivial Stokes data and local monodromy determined  by the parameter $\alpha\in \bC$. 

To obtain a Joyce structure, we must impose that the quadratic differential \( Q_0 \) has zero residue at the pole \( x = \infty \), which corresponds to the condition \( \alpha = 0 \). Consequently, the base \( M \) of the Joyce structure has natural coordinates \( (t, H) \). However, since most of our computations are valid for arbitrary values of $\alpha$,  we will retain it as a  variable for most of the section.

\subsection{Connections}

The relevant Higgs field is the Jimbo-Miwa Lax matrix for Painlev\'e II
\begin{equation}\label{eq:HiggsPII}
	\Phi(x)=x^2\left( \begin{array}{cc}
		1 & 0 \\
		0 & -1
	\end{array} \right)+
	x\left( \begin{array}{cc}
		0 & 1 \\
		t-2p+2q^2 & 0
	\end{array} \right)
	+\left( \begin{array}{cc}
		p-q^2 & -q \\
		-2\alpha+q(t-2p+2q^2) & -p+q^2
	\end{array} \right).
\end{equation}
It is a classical result \cite{JMU1981II} that isomonodromic deformations of the flat connection $\dd-\Phi(x)\dd x$ are described by the second Painlev{\'e} equation 
\begin{equation}\label{p2}
	\frac{d^2q}{dt^2}=2q^3+qt-\alpha+\frac{1}{2}.
\end{equation}

\begin{remark}
    \label{footnote:PII}The matrix
\begin{equation}
	\Phi(x)= \left( \begin{array}{cc}
	1 & 0 \\
	0 & -1		
	\end{array} \right)x^2 
	+ \left( \begin{array}{cc}
		0 & u \\
		-\frac{2z}{u} & 0
	\end{array} \right)x
	+\left( \begin{array}{cc}
		z+\frac{t}{2} & -u y \\
		-\frac{2}{u}(\theta+yz) & -z-\frac{t}{2}
	\end{array} \right)
\end{equation}
of \cite{JMU1981II}  is parameterised in terms of $y,z,u,\alpha,t$. To obtain \eqref{eq:HiggsPII}, we set $u=1$ by a diagonal gauge transformation, and write our Higgs field in terms of spectral Darboux coordinates $\alpha=\theta$, $q=y$ and $p=z+q^2+\half t$.  
\end{remark}

As before, we introduce a parameter $\e\in \bC^*$ and consider  connections $\dd-A_\e(x) \dd x$, where
\begin{equation}
	A_\e (x)=A_\infty+\frac{1}{\e}\Phi(x),\qquad A_\infty=\left(\begin{array}{cc}
		r & 0 \\
		-2s-1-2r(x+q) & -r
	\end{array} \right).
\end{equation}
This is obtained from \eqref{eq:HiggsPII} by shifting
\begin{equation}
    p\mapsto p+\e r,\qquad \alpha\mapsto\alpha+\e\left(s+\frac{1}{2}\right).
\end{equation}
We then define the Hamiltonian
\begin{equation}\label{eq:PIIHam}
    H=\frac{1}{2}\left(p^2-q^4-tq^2+2\alpha q\right),
\end{equation}
and view the connection \eqref{family} as being parameterised by $\e\in \bC^*$ together with the variables $(t,H, \alpha,q,r,s)$, with $p$ defined implicitly by  \eqref{eq:PIIHam}.

\subsection{Extended isomonodromic flows}
We consider isomonodromic flows obtained by varying $(t,H,\alpha)$ and insisting that the generalised monodromy of $\dd-A_\e (x) \dd x$ is constant.

\begin{prop}
    The isomonodromy flows are generated by the vector fields
   \begin{equation}\label{eq:IsoFlow1}
	w_1=\frac{\partial}{\partial t} +  \left(r+\frac{p}{\e}\right)\frac{\partial}{\partial q}-\frac{1}{\e}\left[s+\frac{q^2}{2p} +\frac{r}{p}(2q^3+tq-\alpha)\right]\frac{\partial}{\partial r} ,
\end{equation}
\begin{equation}\label{eq:IsoFlow2}
	w_2=\frac{\partial}{\partial H}-\frac{1}{\e p}\frac{\partial}{\partial r},\qquad w_3=\frac{\partial}{\partial \alpha}-\frac{1}{\e}\frac{\partial}{\partial s}+\frac{q}{ \e p} \frac{\partial}{\partial r}.
\end{equation}
\end{prop}
\begin{proof}
    We first consider the connection as a function of the variables $(t,\alpha,q,p,r,s)$, with $H$ defined by \eqref{eq:PIIHam}. As before, we can describe isomonodromic deformations of $d-A_\e(x)\, dx$ with parameter $t$ as the compatibility conditions \eqref{consistency} of the linear systems \eqref{systems}. This forces
\begin{equation}B_\epsilon(x)=\left( \begin{array}{cc}
        0 & 0 \\
        -r & 0
    \end{array} \right)+
    \frac{1}{2\e}\left(
\begin{array}{cc}
 q+x & 1 \\
 -2 p+2 q^2 +t & -q-x \\
\end{array}
\right),
\end{equation}
and gives the conditions

\begin{equation}
\label{normal}
    \frac{\dd q}{\dd t}=\left(\frac{p}{\e}+r \right), \qquad
    \frac{\dd }{\dd t}(p+\e r)=\frac{2q^3+tq-\alpha}{\epsilon} -s, \qquad  \frac{\dd }{\dd t}(\alpha+\e s)=0.
\end{equation}

The isomonodromic flow is still undetermined, due to the additional  flows 
\begin{equation}
\label{extramore}
    \frac{\partial}{\partial p}-\frac{1}{\e }\frac{\partial}{\partial r}, \qquad \frac{\partial}{\partial \alpha}-\frac{1}{\e }\frac{\partial}{\partial s},
\end{equation}
reflecting the fact that $A_\epsilon(x)$ and $B_\epsilon(x)$ depend only on the combinations $p+\e r$ and $\alpha+\e s$. We now rewrite our flows in the co-ordinates $(t,H,\alpha,q,r,s)$ with $p$ defined by \eqref{eq:PIIHam}. Taking linear combinations of \eqref{extramore} gives  the flows
\begin{equation}
\label{extra}
    {w}_2=\frac{\partial}{\partial H}-\frac{1}{\e p}\frac{\partial}{\partial r}, \qquad w_3=\frac{\partial}{\partial\alpha}-\frac{1}{\e}\frac{\partial}{\partial s}+\frac{q}{ \e p} \frac{\partial}{\partial r}.
\end{equation}
We can add copies of  these flows  to normalise \eqref{normal} in such a  way that $\frac{\dd H}{\dd t}=\frac{\dd \alpha}{\dd t}=0$. This gives the system of equations

\begin{gather}
\label{eq:IsoPII}
  \frac{\dd q}{\dd t}=\left(r+\frac{p}{\e}\right), \qquad 
    \frac{\dd p}{\dd t}=\frac{q^2}{2p}+\left(\frac{1}{\e}+\frac{r}{p} \right)(2q^3+tq-\alpha) , \\
    \frac{ \dd r}{\dd t}=-\frac{1}{\e}\left[s+\frac{q^2}{2p}+\frac{r}{p}\left(2q^3+tq-\alpha \right) \right],
\end{gather}
generated by the flow $w_1$.
\end{proof}
\begin{remark}
    As in the example from the previous section, 
irrespectively of how we normalise the flow in $t$, the resulting nonlinear second order nonlinear ODE for $q$ is the appropriate Painlev\'e equation, in this case the $\e$-deformed Painlev\'e II equation with shifted monodromy parameter
\begin{equation}
    \e^2\frac{\dd^2 q}{\dd t^2}=2q^3+qt-(\alpha+\e s).
\end{equation}
This reduces to \eqref{p2} when $\epsilon=1$ and $s=-\half$.
\end{remark}

\subsection{Spectral curve}

By the gauge transformation \eqref{gauge} we can rewrite \eqref{family} as a scalar equation of the form \eqref{above2}. In this case $Q_0(x)$ is given by \eqref{q0}, and
\begin{equation}\label{eq:Q0Q1PII}
	Q_1(x)=-\frac{p}{x-q}
	+2pr-2s(x-q),
\qquad
	Q_2(x)=\frac{3}{4(x-q)^2}-\frac{r}{x-q}+r^2.
\end{equation}
 
Let us fix the parameters $(t,H,\alpha)$. The corresponding spectral curve $\Sigma$ is a smooth curve of genus 1 with a double cover $p\colon \Sigma\to \bP^1$ branched over 4 points $x_i\in \bC^*$. We denote by $\infty_{\pm}$ the two inverse images of the point $x=\infty$. The open subset $\Sigma^0\subset \Sigma$ is then the complement $\Sigma^0=\Sigma\setminus\{\infty_\pm\}$, and is the  affine quartic  curve
\begin{equation}\label{eq:SigmaPII}
	y^2=x^4+tx^2-2x\alpha+2H.
\end{equation}

Similar reasoning to the proof of Lemma \ref{nak} shows that the inclusion $i\colon \Sigma^0\hookrightarrow \Sigma$ induces a short exact sequence
\begin{equation}
    0\lra \bZ\lra H^1(\Sigma^0,\bZ)^-\stackrel{i_*}{\lra} H^1(\Sigma,\bZ)^-\lra 0,\end{equation}
   where  as before $H_1(\Sigma,\bZ)^-=H_1(\Sigma,\bZ)$. The kernel of $i_*$ is generated by a small anti-clockwise loop $\gamma_3$ around the point $\infty_+$. Note that this class  is anti-invariant because $\sigma^*(\gamma_3)$ is represented by   a small anti-clockwise loop around $\infty_-$ and thus is homologous to $-\gamma_3$.

We take generators $\gamma_1,\gamma_2,\gamma_3$ for $H_1(\Sigma^0,\bZ)^-$ such that $\gamma_1,\gamma_2$ project to  generators of $H_1(\Sigma,\bZ)\isom\bZ^{\oplus 2}$.  
With appropriate ordering of $\gamma_1,\gamma_2$ we then  have
 \begin{equation}
     \<\gamma_1,\gamma_2\>=1, \qquad \<\gamma_1,\gamma_3\>=0=\<\gamma_2,\gamma_3\>.
 \end{equation}

Let us introduce the meromorphic differentials
\begin{equation}\lambda = y \dd x, \qquad \theta=-\frac{Q_1(x) \,dx}{2\sqrt{Q_0(x)} },\end{equation}
and the associated periods
\begin{equation}\label{eq:PeriodII}z_i=\oint_{\gamma_i} \lambda, \qquad \theta_i=\oint_{\gamma_i} \theta.\end{equation}

The differential
 $\lambda$ has a pole of order 4 at the points $\infty_\pm$, and we can order them so that near $\infty_\pm$ we have
\begin{equation}\label{eq:lambdaPIIExp}
	\lambda=\pm\left[x^2+\frac{t}{2}-\frac{\alpha}{x}+\frac{1}{x^2}\left(H-\frac{t^2}{8} \right)+O(x^{-3}) \right]\dd x.
\end{equation}In particular $\lambda$ has residue $\alpha$ at $\infty_+$, so that 
\begin{equation}\label{z3}z_3=\oint_{\gamma_3} \lambda=2\pi i \alpha.\end{equation}

The differential $\theta$ has a simple pole at $\infty_\pm$ with residue $\mp s$, where the leading behaviour is
\begin{equation}
	\theta=\mp\left[-\frac{s}{x}+ \frac{pr+qs}{x^2}+\frac{st-p}{2x^3}+ O(x^{-4}) \right]\dd x.
\end{equation}
It also has   simple poles at the points $x=(q,\pm p)$ with residue $\pm \half$.

As before, defining  the periods $\theta_i$ requires choosing lifts of the cycles $\gamma_i\in H_1(\Sigma^0,\bZ)^-$ to elements $\gamma_i^*\in H_1(\Sigma^*,\bZ)^-$, where $\Sigma^*=\Sigma^0\setminus\{(q,\pm p)\}$. Different choices 
affect the values of $\theta_i$ by the addition of integer multiples of $2\pi i$. Note that lifting $\gamma_3\in H_1(\Sigma,\bZ)^-$  in the most obvious way does not define an anti-invariant cycle on $\Sigma^*$. Instead, we take $\gamma_3^*\in H_1(\Sigma^*,\bZ)^-$ to be the sum of small anti-clockwise loops around the points $\infty_+$ and $(q,p)$.

\subsection{Computation of periods}
Introduce the differentials
\begin{equation}
	\omega=\frac{\partial\lambda}{\partial H}=\frac{\dd x}{y}, \qquad \beta_t=\frac{\partial\lambda}{\partial t}=\frac{x^2\, \dd x}{2y}, \qquad \beta_\alpha=\frac{\partial\lambda}{\partial \alpha}=-\frac{x\, \dd x}{y} .
\end{equation}
Then 
$\omega$ is holomorphic on $\Sigma$, with leading order behaviour  at $\infty_{\pm}$ given by
\begin{equation}
\omega=\pm\left[\frac{1}{x^2}+O(x^{-3}) \right]\dd x.
\end{equation} The differential $\beta_t$ has poles of order 2 at $\infty_\pm$ with zero residue, 
whereas $\beta_\alpha$ has simple poles at $\infty_\pm$  with residue $\pm 1$. Differentiating \eqref{eq:lambdaPIIExp}, the expansions at $\infty_\pm$ are 
\begin{equation}
	\beta_t=\pm\left[\frac{1}{2}-\frac{t}{4x^2} +O(x^{-3}) \right]\dd x,\qquad \beta_\alpha=\pm\left[-\frac{1}{x}+O(x^{-3}) \right]\dd x.
\end{equation}
Since these differentials are regular on $\Sigma^0$ the corresponding periods $\beta_{t,i}$, $\omega_i$, $\beta_{\alpha,i}$ are  well-defined.

Define $\theta_t,\,\theta_H,\,\theta_\alpha$ through the relation
\begin{equation}
    \theta_1\,\frac{\partial}{\partial z_1}+\theta_2\,\frac{\partial}{\partial z_2}+\theta_3\,\frac{\partial}{\partial z_3}=\theta_t\,\frac{\partial}{\partial t}+\theta_H\,\frac{\partial}{\partial H}+\theta_\alpha\,\frac{\partial}{\partial \alpha}.
\end{equation}
Using the obvious relation
\begin{equation}
\label{againy}
dz_i=\frac{\partial z_i}{\partial t} dt + \frac{\partial z_i}{\partial H} dH +\frac{\partial z_i}{\partial \alpha} d\alpha=\beta_{t,i}\, dt+\omega_i\, dH+\beta_{\alpha,i} \, d\alpha ,
\end{equation}it follows that
\begin{equation}\label{hi}
	\theta_i=\theta_t\beta_{t,i}+\theta_H\omega_i+\theta_\alpha\beta_{\alpha,i}.
\end{equation}

\begin{lemma}
\label{74}
There are identities\begin{equation}\label{eq:thetat1}
	\theta_t=\frac{1}{2}\int_{(q,-p)}^{(q,p)}\omega-\frac{1}{2}\int_{\infty_-}^{\infty_+}\omega
    ,  \qquad  \theta_\alpha =\half -s,\end{equation}
    \begin{equation}\label{eq:thetah1}
	\theta_H=-\frac{1}{2}\int_{(q,-p)}^{(q,p)}\beta_t-(pr+qs)-\Res_{x=\infty_+}\left(\frac{dx}{x}\int^x\beta_t \right).
\end{equation}

\end{lemma}

    \begin{proof}Recall that  we can lift the cycles $\gamma_i\in H_1(\Sigma^0,\bZ)^-$ to elements $\gamma_i^*\in H_1(\Sigma^*,\bZ)^-$. The periods $\theta_i$ are then well-defined. As above, we 
     take $\gamma_3^*$ to be  the sum of anti-clockwise cycles round the points $\infty_+$ and $(q,p)$.

    To apply the Riemann bilinear relations we first cut the surface $\Sigma^*$ along the  cycles  $\gamma_1^*,\gamma_2^*$  to obtain a simply-connected fundamental domain. We take all integrals along paths in this region. Using notation as in the proof of Lemma \ref{thm:lemmathetaIII}, we have
\begin{equation}
\label{hand}
	\langle\omega,\beta_t\rangle=2\pi i\sum_{x_i\in \{\infty_{\pm}\}}\Res_{x=x_i}\left(\beta_t\int^x\omega \right)=2\pi i,
\end{equation}
\begin{equation}
    \langle\omega,\,\beta_\alpha\rangle=2\pi i\sum_{x_i\in \{\infty_{\pm}\}}\Res_{x=x_i}\left(\beta_\alpha\int^x\omega \right)=2\pi i\int_{\infty_-}^{\infty_+}\omega,
\end{equation}
\begin{equation}
\label{problem?}
    \langle\beta_t,\,\beta_\alpha\rangle=2\pi i\sum_{x_i\in \{\infty_{\pm}\}}\Res_{x=x_i}\left(\beta_\alpha\int^x\beta_t \right)
    =-4\pi i\,\Res_{x=\infty_+}\left(\frac{dx}{x}\int^x\beta_t \right),
\end{equation}
\begin{equation}
	\langle\omega,\,\theta\rangle=2\pi i\sum_{x_i\in \{(q,\pm p),\infty_{\pm}\}}\Res_{x=x_i}\left(\theta\int^x\omega \right) =\pi i\int_{(q,-p)}^{(q,p)}\omega-s\<\omega,\beta_\alpha\>,
\end{equation}
\begin{equation}
	\langle\beta_t,\,\theta\rangle=2\pi i\sum_{x_i\in \{(q,\pm p),\infty_{\pm}\}}\Res_{x=x_i}\left(\theta\int^x\beta_t \right)=\pi i\int_{(q,-p)}^{(q,p)}\beta_t+2\pi i (pr+qs)-s\<\beta_t,\beta_\alpha\>.
\end{equation}

By our choice of the lift $\gamma_3^*$ we have
 $\theta_3=2\pi i (\half-s)$. Moreover $\omega_3=\beta_{t,3}=0$ and $\beta_{\alpha,3}=2\pi i $ so from \eqref{hi} we get $\theta_\alpha=\half-s$. 
Then
\begin{equation}
	\pi i\int_{(q,-p)}^{(q,p)}\omega-2\pi is\int_{\infty_-}^{\infty_+}\omega =\langle\omega,\,\theta\rangle=\theta_t\langle\omega,\,\beta_t\rangle+\theta_\alpha\langle\omega,\,\beta_\alpha\rangle=2\pi i\theta_t+2\pi i\theta_\alpha\int_{\infty_-}^{\infty_+}\omega,
\end{equation}
which gives  \eqref{eq:thetat1}. Equation \eqref{eq:thetah1} follows in the same way. 
 \end{proof}

\begin{remark}
Note that
\begin{equation}
    \label{thatone}
\theta_t=\frac{1}{2}\int_{(q,-p)}^{(q,p)}\omega-\frac{1}{2}\int_{\infty_-}^{\infty_+}\omega=\frac{1}{2} \int_{(q,-p)}^{\infty_-}\omega+ \frac{1}{2}\int_{\infty_+}^{(q,p)}\omega=\int_{\infty_+}^{(q,p)}\omega\end{equation}
using anti-invariance of $\omega$. 
Taking a base-point $(x_0,y_0)\in \Sigma$ we have

\begin{equation}
\begin{split}\Res_{x=\infty_{+}}\left(\frac{dx}{x}\int^{(x,y)}_{(x_0,y_0)}\beta_t \right) &= \Res_{x=\infty_{+}}\left(\frac{dx}{x}\int^{(x,y)}_{(x_0,y_0)}(\beta_t- \half dx) \right)+  \Res_{x=\infty_{+}}\left(\frac{dx}{x}\int_{x_0}^{x} \half dx \right)\\
&=-\int_{(x_0,y_0)}^{\infty_+} (\beta_t- \half dx) + \half x_0.\end{split}\end{equation}
Choosing $(x_0,0)\in \Sigma$ to be a branch-point we can then use anti-invariance of $\beta_t$ to write
\begin{equation}
 \label{thisone}
\begin{split}\theta_H&=-\int_{(x_0,0)}^{(q,p)} (\beta_t-\half dx)-\half(q-x_0) -(pr+qs)+\int_{(x_0,0)}^{\infty_+} (\beta_t-\half dx) -\half x_0 \\
   &=-\int_{\infty_+}^{(q,p)} (\beta_t-\half dx) -(pr+qs)-\frac{q}{2}.\end{split}\end{equation}
\end{remark}

\subsection{Isomonodromic flows in new coordinates}

 Let $\Xi:\,(t,H,\alpha, q,r,s)\rightarrow(t,H,\alpha,\theta_t,\theta_H,\theta_\alpha)$ be the change of coordinates given by the formulae of Lemma \ref{74}.

\begin{proposition}\label{thm:XiII}
    The push-forward of the isomonodromic flows along $\Xi$ can be written in the form
\begin{equation}
	\Xi_*(w_1)=\frac{\partial}{\partial t}+\frac{1}{\e}\frac{\partial}{\partial \theta_t}+\frac{\partial^2K}{\partial\theta_t\partial\theta_H}\frac{\partial}{\partial\theta_t}-\frac{\partial^2K}{\partial\theta_t^2}\frac{\partial}{\partial\theta_H},
\end{equation}
\begin{equation}
	\Xi_*(w_2)=\frac{\partial}{\partial H}+\frac{1}{\e}\frac{\partial}{\partial \theta_H}+\frac{\partial^2K}{\partial\theta_H^2}\frac{\partial}{\partial\theta_t}-\frac{\partial^2K}{\partial\theta_H\partial\theta_t}\frac{\partial}{\partial\theta_H},
\end{equation}
\begin{equation}\Xi_*(w_3)=\frac{\partial}{\partial \alpha}+\frac{1}{\e}\frac{\partial}{\partial \theta_\alpha}+\frac{\partial^2K}{\partial\theta_\alpha\partial \theta_H}\frac{\partial}{\partial\theta_t}-\frac{\partial^2K}{\partial\theta_\alpha\partial\theta_t}\frac{\partial}{\partial\theta_H}.\end{equation}
\end{proposition}

\begin{proof}
For $i=0,\ldots, 4$ we define
\begin{equation}
	\kappa_i:=\frac{1}{4}\int_{(q,-p)}^{(q,p)}\frac{x^i\dd x}{y^3}+\frac{1}{2}\Res_{\zeta=\infty_+}\left(\frac{d\zeta}{\zeta}\int^\zeta \frac{x^idx}{y^3} \right).
\end{equation} 
Differentiating the change of coordinates gives
\begin{equation}
	\Xi_*\left(\frac{\partial}{\partial t} \right)=\frac{\partial}{\partial t}-\kappa_2\frac{\partial}{\partial\theta_t}+\frac{1}{2}\left(\kappa_4-\frac{rq^2}{p} \right)\frac{\partial}{\partial \theta_H},
\end{equation}
\begin{equation}
	\Xi_*\left(\frac{\partial}{\partial H} \right)=\frac{\partial}{\partial H}-2\kappa_0\frac{\partial}{\partial\theta_t}+\left(\kappa_2-\frac{r}{p} \right)\frac{\partial}{\partial\theta_H},
\end{equation}
\begin{equation}
	\Xi_*\left(\frac{\partial}{\partial\alpha} \right)=\frac{\partial}{\partial\alpha}+2\kappa_1\frac{\partial}{\partial\theta_t}-\left(\kappa_3-\frac{rq}{p}\right)\frac{\partial}{\partial\theta_H},
\end{equation}
\begin{equation}
	\Xi_*\left(\frac{\partial}{\partial r} \right)=-p\frac{\partial}{\partial\theta_H},\qquad \Xi_*\left(\frac{\partial}{\partial s} \right)=-\frac{\partial}{\partial\theta_\alpha}-q\frac{\partial}{\partial\theta_H},
\end{equation}
\begin{equation}
	\Xi_*\left(\frac{\partial}{\partial q} \right)=\frac{1}{p}\frac{\partial}{\partial\theta_t}-\left(\frac{r}{p}(2q^3+tq-\alpha)+s+\frac{q^2}{2p}\right)\frac{\partial}{\partial\theta_H}.
\end{equation}
We can invert this to find
\begin{equation}
\label{pegg}
	\Xi_*^{-1}\left(\frac{\partial}{\partial\theta_t} \right)=p\frac{\partial}{\partial q}-\left(\frac{r}{p}(2q^3+tq-\alpha)+s+\frac{q^2}{2p}\right)\frac{\partial}{\partial r},
\end{equation}
\begin{equation}
\label{pegg2}
	\Xi_*^{-1}\left(\frac{\partial}{\partial\theta_H} \right)=-\frac{1}{p}\frac{\partial}{\partial r}, \qquad \Xi_*^{-1}\left(\frac{\partial}{\partial\theta_\alpha} \right)=-\frac{\partial}{\partial s}+\frac{q}{p}\frac{\partial}{\partial r}.
\end{equation}

The isomonodromy flows become
\begin{equation}
 \Xi_*(w_1)  = \frac{\partial}{\partial t}+\frac{1}{\e}\frac{\partial}{\partial\theta_t}-\left(\kappa_2-\frac{r}{p} \right)\frac{\partial}{\partial\theta_t}+\left(\frac{\kappa_4}{2}-\frac{q^2r}{p}-rs-\frac{r^2}{p}(2q^3+tq-\alpha)\right)\frac{\partial}{\partial\theta_H},
\end{equation}
\begin{equation}
    \Xi_*(w_2) = \frac{\partial}{\partial H}+\frac{1}{\e}\frac{\partial}{\partial\theta_H}-2\kappa_0\frac{\partial}{\partial\theta_t}+\left(\kappa_2-\frac{r}{p} \right)\frac{\partial}{\partial\theta_H},
\end{equation}
\begin{equation}
    \Xi_*(w_3)=\frac{\partial}{\partial\alpha} +\frac{1}{\e}\frac{\partial}{\partial\theta_\alpha} +2\kappa_1\frac{\partial}{\partial\theta_t}-\left(\kappa_3-\frac{rq}{p}\right)\frac{\partial}{\partial\theta_H},
\end{equation}
and arguing as in the proof of Proposition \ref{propIII}  gives the existence of the required function $K$.
\end{proof}

The function $K$ satisfies
\begin{equation}\label{herehere1}
    \frac{\partial^3K}{\partial\theta_H^3}=0, \qquad \frac{\partial^3K}{\partial\theta_H^2\partial\theta_t}=-\frac{1}{p^2},\qquad \frac{\partial^3K}{\partial\theta_H\partial\theta_t^2}=-\frac{2r \left(2 q^3+ tq -\alpha\right)}{p^2}-\frac{q^2}{p^2}-\frac{s}{p},
\end{equation}
\begin{multline}\label{herehere2}
    \frac{\partial^3K}{\partial\theta_t^3}=-\frac{3 q^4}{4 p^2}+r \left(2q-\frac{3 q^2 \left(2q^3+tq-\alpha \right)}{p^2}\right)\\+  r^2 \left(6 q^2+t-\frac{3 \left(2q^3+tq-\alpha \right)^2}{p^2}\right)
    -\frac{3sq^2}{2p}-\frac{3sr\left(2q^3+tq-\alpha\right)}{p}-s^2,
\end{multline}

\begin{equation}
    \frac{\partial^3K}{\partial\theta_\alpha \partial \theta_H^2}=0, \qquad \frac{\partial^3K}{\partial\theta_\alpha \partial \theta_H \partial \theta_t}=\frac{q}{p^2}, \qquad \frac{\partial^3K}{\partial\theta_\alpha \partial \theta_t^2}=\frac{q^3}{p^2}-r+\frac{sq}{p} +\frac{2rq}{p^2} (2q^3+tq-\alpha),
\end{equation}
\begin{equation}
   \frac{\partial^3K}{\partial\theta_\alpha^2 \partial \theta_H}=0, \qquad \frac{\partial^3K}{\partial\theta_\alpha^2 \partial \theta_t}=-\frac{q^2}{p^2}. 
\end{equation}

\subsection{Pleba{\'n}ski function}

We explained in Section \ref{sec:Quad} that to obtain a Joyce structure we must restrict to the symplectic leaf \eqref{m} obtained by setting the residue $\alpha=0$. Equivalently, by \eqref{z3}, this is  the locus $z_3=0$. Recall from \eqref{poisson} that the Poisson structure  is given by \begin{equation}\{z_1,z_2\}=2\pi i \, \<\gamma_1,\gamma_2\>=2\pi i.\end{equation} Inverting this gives the symplectic form  
\begin{equation}\label{handyII} \Omega_0=-\frac{1}{2\pi i} \cdot dz_1\wedge dz_2 = \frac{1}{2\pi i}  (\omega_1\beta_{t,2}-\omega_2\beta_{t,1})\cdot dt\wedge dH=dt\wedge dH,\end{equation}
on $M$, where we used  \eqref{againy} and \eqref{hand}.

We must also  restrict the value of $\theta_3=2\pi i \theta_\alpha$. As discussed in Remark \ref{earlyish} this is a slightly subtle point. To get a Joyce structure it is necessary to take $\theta_3\in \pi i \bZ$. In what follows we shall make the choice $\theta_3=\pi i$, which by Lemma \ref{74} corresponds to setting $s=0$.

\begin{theorem}
Take $\alpha=0=s$. Then in terms of the co-ordinates $(z_1,z_2,\theta_1,\theta_2)$	the isomonodromic flows  take the form
\begin{equation}
\label{eq:thmPlebII}
	\frac{\partial}{\partial z_i}+\frac{1}{\e}\cdot \frac{\partial}{\partial\theta_i}+2\pi i \cdot \frac{\partial^2 W}{\partial\theta_i\partial\theta_1}\cdot \frac{\partial}{\partial\theta_2}- 2\pi i \cdot \frac{\partial^2 W}{\partial\theta_i\partial\theta_2}\cdot \frac{\partial}{\partial\theta_1},
\end{equation}	
where the Pleba{\'n}ski function $W$ is given by
\begin{equation}
\label{plebII}
    W=\frac{p}{48H(t^2-8H)}\left(-tq - 2r(2t^2+3q^2t-12H) +12r^2q(-t^2-q^2t+4H)-8r^3p^2t\right).
\end{equation}
\end{theorem}

\begin{proof}
For a fixed point $(t,H)\in M$ we can   adapt the argument of Lemma \ref{abhol} to show that the  map
    \begin{equation}\label{map2} \big\{(q,p)\in (\bC^*)^2: p^2=Q_0(q)\big\}\times \big\{r\in \bC\big\}\to \big\{(e^{\theta_1},e^{\theta_2})\big\}\in (\bC^*)^2\end{equation}
    defined by the equations \eqref{hi}, \eqref{eq:thetat1} and \eqref{eq:thetah1} is  an open embedding. We can therefore view $K$ as a local function of $(z_1,z_2,\theta_1,\theta_2)$. Applying the operators \eqref{pegg} - \eqref{pegg2} to \eqref{herehere1} - \eqref{herehere2} and to \eqref{plebII} shows that the  fourth derivatives of $K$ and $W$ with respect to the theta variables coincide. We can then argue as in the proof of Theorem \ref{thm:Pleba{'n}skiIII}.
\end{proof}

\begin{remark}
 We can also integrate the fourth derivatives of $K$ for arbitrary values of the variables $\alpha,s$. It takes the form
\begin{equation}
    W= \frac{1}{\Delta}\cdot \sum_{k=0}^3\sum_{m=0}^{3-k}W_{k,m}(q)r^k s^m.
\end{equation}
The functions $W_{k,m}(q)$ are as follows:
\begin{equation}
    W_{0,0}=
    p \left[\frac{2}{3}c_1 q +2 c_2\right],\quad W_{0,1}=
q \left[4c_1q^2 +4 c_2  q +\frac{4}{3} \left(96 H^2+4 H t^2-2 t^4-27 \alpha ^2 t\right)\right],
\end{equation}
\begin{equation}
    W_{0,2}=
    8p \left[c_1q -c_2\right],\quad     W_{0,3}=
    16 q \left[\frac{c_1}{3}q^2 -c_2 q+4 H t^2-3 \alpha ^2 t -32 H^2\right],
\end{equation}
\begin{equation}
    W_{1,0}=
    4p\left[c_1q^2 +2c_2  q +\frac{1}{3} \left(-96 H^2+28 H t^2-2 t^4-45 \alpha ^2 t\right) \right],\quad W_{1,1}=
16c_1p^2, 
\end{equation}
\begin{equation}
    W_{1,2}=
    16p\left[c_1q^2-2c_2  q +4 H t^2-3 \alpha ^2 t-32 H^2 \right],
\end{equation}
\begin{equation}
    W_{2,0}=
    8p\left(c_1 q^3+c_2 q^2+q \left(-32 H^2+12 H t^2-21t \alpha ^2-t^4\right)+\alpha  \left(27 \alpha ^2-24 H t+t^3\right) \right),
\end{equation}
\begin{equation}
    W_{2,1}=
    16p^2\left(c_1 q-c_2 \right),\qquad     W_{3,0}=
    \frac{16p^3}{3}c_1,
\end{equation}
where 
\begin{equation}
    c_1:=-t^3+8Ht-18\alpha^2,\quad c_2:=-\alpha(t^2+24H),
\end{equation}
\begin{equation}
    \Delta:=16 \left(-27 \alpha ^4-\alpha ^2 t \left(t^2-72 H\right)+2 H \left(t^2-8 H\right)^2\right).
\end{equation}

\end{remark}

\subsection{Behaviour on the zero section via uniformisation}

Recall that we have set $\alpha=0=s$. Introducing the variables 
\begin{equation}
    \label{haha}X=\frac{t}{12}+\frac{1}{2}(y+x^2), \qquad Y=\frac{tx}{2}+xy+x^3,\end{equation}
the spectral curve \eqref{eq:SigmaPII}  takes the standard Weierstrass form
\begin{equation}Y^2=4X^3-g_2 X -g_3, \qquad g_2=\frac{1}{12}(24H+t^{2}), \qquad g_3=\frac{ t}{216}(72H -t^{2}).
\end{equation}
The inverse transformation to \eqref{haha} is
\begin{equation}\label{trees}x = \frac{3Y}{6X+t}, \qquad y =  \frac{144X^2+48Xt+72H-5t^2}{24(6X+t)}.\end{equation}

We can uniformise $\Sigma$ by writing $X=\wp(u)$ and $Y=\wp'(u)$, with $\wp(u)$ the Weierstrass $\wp$-function for the lattice spanned by the periods $\omega_i$. Beware that the involution $u\mapsto -u$ does not correspond to $\sigma\colon (x,y)\mapsto (x,-y)$ but rather to $\sigma'\colon (x,y)\mapsto (-x,y)$. A calculation shows that
\begin{equation}\omega=\frac{dx}{y}=\frac{dX}{Y}=du, \qquad \beta_t=\frac{x^2\omega}{2} = \frac{X\, dX}{Y}-\frac{t\, dX}{12Y} -\frac{dx}{2}.\end{equation}

Recall from Section \ref{summy} the definition of the function $S=S(t,H)$ and the linear Joyce connection.


\begin{theorem}\label{thm:Pleba{'n}skiOperII}
\begin{itemize}
    \item[(i)]The Pleba{\'n}ski function $W$ is regular on the locus $\theta_t=\theta_H=0$.
    
    \item[(ii)] There is an identity
     \begin{equation}S= \log \left(H^2(8H-t^2)\right)^{-\frac{1}{48}}.\end{equation}
 \item[(iii)] The co-ordinates $(t,H-\tfrac{1}{8}
 t^2)$ on $M$ are flat for the linear Joyce connection.
 \end{itemize}
\end{theorem}

\begin{proof}
Recall that we have set $\alpha=0=s$. We denote by $v\in \bC$ a point corresponding to the point $(q,p)\in \Sigma$.  Using the involution $\sigma'$ the formula \eqref{thisone} becomes 
       \begin{equation}\label{shoulder}\theta_H=-\frac{1}{2}\int_{(-q,p)}^{(q,p)}(\beta_t-\half dx) -pr-\frac{q}{2}=\zeta(v) +\frac{tv}{12}+\frac{q}{2}-pr.\end{equation}
       Using the expansion
\begin{equation}\zeta(v)=\frac{1}{v}-\left(\frac{H}{30}+\frac{t^2}{720}\right)v^3 +O(v^5),\end{equation}
and introducing a new parameter $w$ by writing \begin{equation}\label{rr}r=\frac{v}{2} - wv^2+\frac{tv^3}{12}\end{equation} we then have
\begin{equation}\label{care}\theta_t=v, \qquad \theta_H=w+\frac{twv^2}{6} +O((v,w)^3).\end{equation}

The formulae \eqref{trees} and the expansion of $\wp(v)=-\zeta'(v)$  gives
\begin{equation}
    q= -\frac{1}{v}+\frac{tv}{6}+\left(\frac{H}{5}-\frac{7t^2}{360}\right) v^3+O(v^5), \qquad p=\frac{1}{v^2}+\frac{t}{6}+\left(\frac{3H}{5}-\frac{7t^2}{120}\right) v^2 + O(v^4).
\end{equation}
Substituting these expressions and \eqref{rr} into \eqref{plebII} we find that $W$ is regular  along the locus $\theta_t=\theta_H=0$ with leading order behaviour 
\begin{equation}W=\frac{tv}{24(8H-t^2)}\ - \frac{(12H-t^2)w}{24H(8H-t^2)} + O((v,w)^3).\end{equation}
This gives (i) and (ii). For (iii) we substitute \eqref{care} into \eqref{herehere1} - \eqref{herehere2} to get
\begin{equation}\frac{\partial^3 K}{\partial \theta_t^3}\bigg\vert_{\theta_t=\theta_H=0}=-\frac{1}{4}, \qquad \frac{\partial^3 K}{\partial \theta_t^2\theta_H}\bigg\vert_{\theta_t=\theta_H=0}=\frac{\partial^3 K}{\partial \theta_H^3}\bigg\vert_{\theta_t=\theta_H=0}=\frac{\partial^3 K}{\partial \theta_t\partial \theta_H^2}\bigg\vert_{\theta_t=\theta_H=0}=0.\end{equation}
  It follows that the linear Joyce connection $\nabla^J$  satisfies
  \begin{equation}
      \nabla^J_{\frac{\partial}{\partial t}}\left(\frac{\partial}{\partial t}\right)=-\frac{1}{4}\frac{\partial}{\partial H}, \qquad \nabla^J_{\frac{\partial}{\partial H}}\left(\frac{\partial}{\partial t}\right)=0, \qquad \nabla^J_{\frac{\partial}{\partial t}}\left(\frac{\partial}{\partial H}\right)=\nabla^J_{\frac{\partial}{\partial H}}\left(\frac{\partial}{\partial H}\right)=0,
  \end{equation}
  and a basis of flat vector fields is therefore given by $(\partial_t+\tfrac{1}{4} t\partial_H,\partial_H)$. The corresponding flat co-ordinates are $(t,H- \tfrac{1}{8}t^2)$ as claimed.
\end{proof}

\subsection{Alternative approach: poles of Painlev\'e II}
In the same way as for the case of Painlev\'e III, we can alternatively study the limit $\theta_t=\theta_H=0$ by studying the behaviour of the Pleba\'nski function at a pole $t=t_0$ of the Painlev\'e II equation
\begin{equation}
	\e^2\frac{d^2q}{dt^2}=2q^3+qt-\alpha.
\end{equation}
We set $s=0$ but keep $\alpha$ as a parameter for now. Choosing one of the two branches at infinity we find the following behaviour 
\begin{equation}
	q=-\frac{\e }{t-t_0}+\frac{t_0 (t-t_0)}{6 \e }+\alpha\frac{(t-t_0)^2 }{4 \e ^2}+
    q_0 (t-t_0)^3
\end{equation}
\begin{equation}
p=\frac{\e ^2}{(t-t_0)^2}+\frac{t_0}{6}+O(t-t_0), \qquad 
	r=c (t-t_0)^2+\frac{ (t-t_0)}{2 \hbar },
\end{equation}
with the parameters $c,q_0$ left undetermined by the equations. 
We have
\begin{equation}
	\theta_t=\frac{t-t_0}{\e }+\frac{(t-t_0)^4}{12 \e ^3}+O(t-t_0)^5,\qquad \theta_H=-\e^2 c+\frac{(t-t_0)^2}{8 \e }+O(t-t_0)^3,
\end{equation}
so that $t-t_0$ and $c$ parametrise, at leading order, the deviations from the zero section, which is parametrised by $t_0$ and $H_0:=H\left(q(t_0),p(t_0),t_0 \right)$. For $\alpha\ne0$, the function $W$ does not vanish for $\theta=0$. It is given by
\begin{equation}
    W_0(H_0,t_0,\alpha)=\frac{\alpha\,  t_0}{24H_0(8H_0-t_0^2)}
\end{equation}
In particular it is regular along $\theta=0$. When $\alpha=0$, $W$ is associated to a Joyce structure, and
\begin{equation}
	 W\big|_{\alpha=s=0}\simeq  \frac{t_0^2-12 H_0}{24 H_0 \left(8 H_0-t_0^2\right)}\theta_H+\frac{t_0}{24 H_0 \left(8 H_0-t_0^2\right)}\theta _t,
\end{equation}
so that, analogously to what happened in the Painlev\'e III and I cases, the function vanishes linearly. The generating function is $S=\log\left(H^2(8H-t^2) \right)^{-\frac{1}{48}}$, matching with the result from uniformization.

We could do the same compuation by expanding instead arouns $s=1/2$, i.e. $\theta_\alpha=0$. The resulting expression would have been
\begin{equation}
	 W\big|_{\alpha=0,\,\theta_\alpha=0}\simeq  \frac{t_0^2-12 H_0}{24 H_0 \left(8 H_0-t_0^2\right)}\theta_H+\frac{t_0}{24H_0(8 H_0- t_0^2)}\theta _t,
\end{equation}
which has generating function $S=-\frac{1}{24}\log\left(H(8H-t^2)^2 \right)=\log\Delta^{-\frac{1}{24}} $.


\section{Tau functions}
\label{taufun}

In this section we recall from \cite{JT} the definition of the tau function associated to a Joyce structure. This definition has a rather experimental flavour and depends on the choice of certain extra data. In our two examples there are natural choices for this extra data, and we show that the resulting Joyce structure tau function produces a particular normalisation of the corresponding Painlev{\'e} tau function. In the final part we discuss relations with  topological string partition functions. 

\subsection{Joyce structure tau function}
\label{tauf}

We begin with a brief summary of the Joyce structure tau function, referring the reader to \cite{JT} for further details.
Given a Joyce structure on a complex manifold $M$, there is a complex \hk structure on $X=T_M$, and associated closed holomorphic 2-forms 
\begin{equation}
\label{simple}\Omega_0=\Omega_J+i\Omega_K=\half \sum_{p,q} \omega_{pq} \cdot dz_p \wedge dz_q, \qquad 2i \Omega_I=-\sum_{p,q} \omega_{pq} \cdot d \theta_p  \wedge d z_q,\end{equation}
\begin{equation}
\label{w}\Omega_\infty=\Omega_J-i\Omega_K=\half\sum_{p,q}\omega_{pq} \cdot d\theta_p\wedge d\theta_q+\sum_{p,q} \frac{\partial^2 W}{\partial \theta_p \partial\theta_q} \cdot d \theta_p \wedge d z_q +\sum_{p,q} \frac{\partial^2 W}{\partial z_p \partial \theta_q} \cdot dz_p\wedge dz_q,\end{equation}
where $(\omega_{pq})_{p,q=1}^n$ is the inverse matrix to $(\eta_{pq})_{p,q=1}^n$. For each $\epsilon\in \bC^*$ the combination
\begin{equation}
    \Omega_\epsilon=\epsilon^{-2}\Omega_0+2i\epsilon^{-1} \Omega_I+\Omega_\infty
\end{equation}
descends to the twistor fibre $Z_\epsilon$, which is the space of leaves of the flows \eqref{above}. 
If we define the Euler vector field on $X^\hash$ by
\begin{equation}
    \label{euler1}
  E=\sum_i z_i \cdot \frac{\partial}{\partial z_i}  
\end{equation}
then the homogeneity property \eqref{re2} gives rise to relations
\begin{equation}
\label{po}
    d i_E(\Omega_0)=2\Omega_0,\qquad d i_E(\Omega_I)=\Omega_I, \qquad di_E(\Omega_\infty)=0.
\end{equation}

Let us  choose symplectic potentials
\begin{equation}\label{pot}d\Theta_0=\Omega_0, \qquad d\Theta_{1}=\Omega_1,\qquad d\Theta_\infty=\Omega_\infty,\end{equation}
and define $\Theta_I=i_E(\Omega_I)$, so that by \eqref{po} we also have $d\Theta_I=\Omega_I$. 
Then locally on $X$ we can define a function $\tau$, unique up to multiplication by a constant scalar, by  writing
\begin{equation}
\label{tau}d\log(\tau)=\Theta_0+2i\Theta_I +\Theta_\infty-\Theta_1,\end{equation}

This definition is of course vacuous without some procedure for defining the symplectic potentials \eqref{pot} on the various twistor fibres. Although a complete solution to this problem is not yet known,  a  rough recipe  is discussed in \cite{JT}, and it is then interesting to see how this works out in the two examples of Joyce structures constructed in this paper.

\subsection{Computation of $\tau$ in the examples}
\label{sec:deftau}
In the examples of Joyce structures considered in this paper there are some  natural choices for the symplectic potentials \eqref{pot} above. Once these are made we can make the defining relation \eqref{tau}  explicit, and compare the resulting Joyce structure tau function with the corresponding Painlev{\'e} tau function. The relevant choices are as follows:

\begin{itemize}
    
\item[(i)]
 The base $M$ can be identified with an open subset of the cotangent bundle of the one-dimensional manifold  of isomonodromic times. There is then a canonical Liouville form $\lambda$ on $M$ satisfying $d\lambda=-\Omega_0$. Using \eqref{po} we can  then take
\begin{equation}\Theta_0=i_E(\Omega_0)+\lambda.\end{equation}

\item[(ii)]
 The twistor fibre $Z_\epsilon$ with $\epsilon\in \bC^*$   admits canonical systems of logarithmic Fock-Goncharov co-ordinates $(x_1,x_2)$ satisfying
\begin{equation}\Omega_\epsilon=\omega_{12} \cdot dx_1\wedge dx_2.\end{equation}
Having chosen one of these co-ordinate systems we can then set
\begin{equation}\Theta_\epsilon=\omega_{12} \cdot x_1 dx_2.\end{equation}

\item[(iii)]
When $\epsilon=\infty$ the function $r\colon X^\hash\to \bC$ is constant along the flows \eqref{above} and hence  descends to a function $r\colon Z_\infty\to \bC$. The locus  $\{r=0\}\subset Z_\infty$ is then Lagrangian for the symplectic form $\Omega_\infty$. It follows that if we restrict  $\tau$ to the subset $\{r=0\}\subset X^\hash$ then we can simply take $\Theta_\infty=0$.  Note that setting $r=0$ in \eqref{con2} trivialises the $\epsilon$-independent part of the connection $\nabla$ and thus reduces the isomonodromy flows to those normally considered in the Painlev{\'e} literature.
\end{itemize}

 \subsection{Painlev{\'e} III$_3$ case}
In this case we have  $\eta_{12}=4\pi i$ and $\omega_{12}=-1/4\pi i$. From  \eqref{handy} we have
\begin{equation}\Omega_0=-\frac{1}{4\pi i}\,  dz_1\wedge dz_2 =ds\wedge dH,\end{equation}
Similarly, using \eqref{eq:thetatHdef} we find\begin{equation}2i\Omega_I=\frac{1}{4\pi i} (dz_1\wedge d\theta_2 - dz_2\wedge d\theta_1)=  dH\wedge d\theta_s-ds\wedge d\theta_H,\end{equation}
where we used the relations
\begin{equation}\beta_1\frac{\partial \beta_2}{\partial H}-\beta_2\frac{\partial \beta_1}{\partial H} = \omega_1\frac{\partial \beta_2}{\partial s} - \omega_2\frac{\partial \beta_1}{\partial s}, \qquad \beta_1\frac{\partial \omega_2}{\partial H} - \beta_2\frac{\partial \omega_1}{\partial H} = \omega_1\frac{\partial \omega_2}{\partial s} - \omega_2\frac{\partial \omega_1}{\partial s}\end{equation}
to cancel the additional terms. These relations in turn follow by differentiating \eqref{l1} and noting that the definition of the forms $\omega$ and $\beta$ implies
\begin{equation}
\frac{\partial \omega_i}{\partial s}=\frac{\partial \beta_i}{\partial H}.
\end{equation}

Using the formulae \eqref{ths} we now find

\begin{equation}\begin{split}2i\Omega_I&=-\frac{r(tq^2-1)}{q^3p}\,  dH\wedge ds +2qp\, ds\wedge dr+\frac{1}{2q^2p}\, dH\wedge dq+\frac{2r(tq^2-1)+tq^2}{2pq^3}\,  ds\wedge dq \\
&=-dq\wedge dp +2qp\, ds\wedge dr +\varpi,\end{split}
\end{equation}
where the 2-form $\varpi$ vanishes when $r=0$.  

There is a natural cotangent bundle structure on $M$ for which $\rho\colon M\to B$ is the projection to the Painlev{\'e} time $t$. The associated Liouville form is $\lambda =  H ds$. Following the recipe from Section \ref{sec:deftau} 
we get
\begin{equation}
\begin{split}
    \Theta_0 & =i_E(\Omega_0)+\lambda=4 dH-H ds, \\
    2i\Theta_I & = i_E(2i\Omega_I)=2q \, dp +3p\, dq+8qp\, dr+i_E(\varpi),
\end{split}
\end{equation}
where we used the formula \eqref{euler} for the Euler vector field.

The Lagrangian in the twistor fibre $Z_\infty$ is obtained by setting $r=0$. Denoting this locus by $Y^\hash\subset X^\hash$ we get
\begin{equation}
\begin{split}
d\log(\tau|_{Y^\hash})&=\left(4 \, dH-H \, ds\right) +\left(2q \, dp +3p\,dq\right) +\frac{1}{4\pi i}\, x_1 dx_2 \\
&=- H \,d s+ p\, d q  + d\left(4H+2qp\right)+\frac{1}{4\pi i}\, x_1 dx_2.\end{split}\end{equation}
It follows that the restriction of the Joyce structure tau function to the locus $r=0$ can  be identified with the Painlev\'e III$_3$ tau function, since its logarithmic derivative, up to an exact form and a monodromy-dependent normalization, is the classical action differential \cite{Its2018}.  The Joyce structure tau function can be then thought as an extension of the Pailev\'e tau function to an $\e$-deformed isomonodromic problem with a reference connection. As can be expected, when we restrict to the isomonodromic flows, so that everything depends on $t$ only, we have
\begin{equation}
    \partial_t\log(\tau|_{Y^\hash})=H.
\end{equation}

\subsection{Painlev{\'e} II case}
Recall that to define the Joyce structure we have set $\alpha=0=s$. 
In this case we have  $\eta_{12}=2\pi i$ and $\omega_{12}=-1/2\pi i$. From \eqref{handyII} 
\begin{equation}\Omega_0=-\frac{1}{2\pi i}\, dz_1\wedge dz_2 = dt\wedge dH,\end{equation}
Similarly,  we find\begin{equation}2i\Omega_I=\frac{1}{2\pi i} (dz_1\wedge d\theta_2 - dz_2\wedge d\theta_1)=  dH\wedge d\theta_t-dt\wedge d\theta_H,\end{equation}
where we  cancelled the additional terms in the same way as in the previous section. 

Using Lemma \ref{74} we now find
\begin{equation}\begin{split}2i\Omega_I&=\frac{r}{p}\, dt\wedge dH +p\, dt\wedge dr+\frac{1}{p}\, dH\wedge dq+\left(\frac{r}{p}(2q^3 + tq)+\frac{q^2}{2p}\right) dt\wedge dq \\
    &=r\, dt\wedge dp -dq\wedge dp+p \, dt\wedge dr.\end{split}
\end{equation}
The Euler vector field is
\begin{equation}E=\frac{4H}{3}\frac{\partial}{\partial H} +\frac{2t}{3}\frac{\partial}{\partial t} +\frac{q}{3}\frac{\partial}{\partial q}-\frac{r}{3}\frac{\partial}{\partial r}\end{equation}
and $p$ scales with weight $\tfrac{2}{3}$.

There is a natural cotangent bundle structure on $M$ for which $\rho\colon M\to B$ is the projection to the Painlev{\'e} time $t$. The associated Liouville form is $\lambda =  H dt$. Then

\begin{gather}\Theta_0=i_E(\Omega_0)+\lambda=\frac{1}{3}\left(2t\,  dH- H dt\right), \\
 2i\Theta_I= i_E(2i\Omega_I)=\frac{1}{3}\left(-pr\, dt +2p \, dq+(2rt-q)dp+2pt\,dr\right).\end{gather}

The Lagrangian in the twistor fibre $Z_\infty$ is obtained by setting $r=0$. Denoting this locus by $Y^\hash\subset X^\hash$ we get 
\begin{equation}\label{eq:PIItau}\begin{split}
d\log(\tau|_{Y^\hash})&=\frac{1}{3}\left(2t\, dH- H dt \right)+\frac{1}{3}\left(2p\, dq - q\, dp\right) +\frac{1}{2\pi i}\, x_1 dx_2 \\
&=- H d t+ p d q  + \frac{1}{3}\, d\left(2tH-qp\right)+\frac{1}{2\pi i}\, x_1 dx_2.\end{split}\end{equation}

Again, for $r=0$ this is the classical action differential up to an exact form and a monodromy-dependent normalization \cite{Its2016,Its2018}. In particular, we obtain the same expression as \cite{Its2016} up to a factor of 2 due to a rescaling of the symplectic form between our conventions and theirs. When we restrict to the isomonodromic flows we have 
\begin{equation}
    \partial_t\log(\tau|_{Y^\hash})=H.
\end{equation}

\subsection{Topological string partition functions}

In this section we briefly discuss how our work relates to the literature on isomonodromic deformations, class $S$ theories, and topological string theory. Readers interested only in the mathematical aspects of Joyce structures and Painlev\'e equations may safely skip it.

Following \cite{Gamayun2012}, it was soon realised \cite{Bonelli2017} that expansions of Painlev\'e tau functions around critical points reproduce Nekrasov partition functions of corresponding four-dimensional class $S$ theories on the self-dual Omega-background \cite{Nekrasov2003}. This correspondence has been rigorously proven for Fuchsian systems on the Riemann sphere and on the torus \cite{Gavrylenko2016b,DMDG2020}, as well as for many non-Fuchsian cases, including our Painlev\'e III$_3$ example, via Fredholm determinant representations \cite{Gavrylenko2017,Cafasso2017,Gavrylenko2018b}. In the case of Painlev\'e I and II,  the class $S$ partition function can only be computed recursively, but its asymptotic expansion agrees with that of the Painlev\'e tau function to all known orders \cite{Bonelli2017}. For the homogeoneous Painlev\'e II equation  relevant to the  Joyce structure of Section \ref{sec:two} a Fredholm determinant expression is available  \cite{Desiraju2021},  but no explicit combinatorial expansion is known.

It was shown by explicit computation in several examples \cite{Eguchi2003,Taki2007}  that four-dimensional Nekrasov partition functions are degenerations of partial resummations of topological string partition functions on toric Calabi-Yau threefolds. This naturally leads to the identification of isomonodromic tau functions with  nonperturbative completions of topological string partition functions. This view is supported by \cite{Bonell2017b}, which argues that a nonperturbative topological string partition function for local $\mathbb{P}^1\times\mathbb{P}^1$, defined through the TS/ST correspondence \cite{Grassi2014}, coincides with the Painlev\'e III$_3$ tau function in an appropriate scaling limit.

This correspondence was related to moduli spaces of quadratic differentials in \cite{Coman2022,Coman2020}, where isomonodromic tau functions were proposed to be identified with B-model partition functions of non-compact Calabi-Yau threefolds given locally by equations of the form 
\begin{equation}\label{eq:ClassSigma}
    uv+y^2=Q_0(x)\subset \mathbb{C}^4.
\end{equation}

Although previous computations were based either on direct matching with Nekrasov partition functions, topological vertex, or conjectural string dualities, a direct relation between the B-model topological string partition functions on \eqref{eq:ClassSigma} and Painlev\'e tau functions was recently established in \cite{Bonelli2024}. There it was shown that suitably normalised Painlev\'e tau functions are holomorphic and modular solutions to the holomorphic anomaly equations \cite{BCOV} characterising B-model topological string free energies.

In our examples, we verify that an appropriate specialisation of the tau functions associated with Joyce structures (as introduced in \cite{JT}) coincides with the Painlev\'e II and III$_3$ tau functions. This supports the view that the tau function of a Joyce structure extends the conventional isomonodromic tau functions to cases with a nontrivial reference connection, and is naturally related to a nonperturbative topological string partition function. This adds evidence to the two known cases up to now, which are the Joyce structures associated to Painlev\'e I and to the derived category of coherent sheaves on the resolved conifold.

\begin{remark}\label{rmk:SF}
The expression \eqref{hi2} also  appears in the context of topological string theory on non-compact Calabi-Yau threefolds given locally by an equation of the form $y^2+uv=Q_0(x)$. More precisely one has $S\propto F_1$, where $F_1$ is the first correction to the refined B-model topological string free energy in the so-called Nekrasov-Shatashvili (NS) limit (see \cite[equation (3.23)]{Huang2011}).
The appearance of the NS free energy in the locus $\theta=0$ is particularly suggestive: as we show in the text, this locus corresponds to the poles of the Painlev\'e equation, or equivalently to the zeros of the tau function. Approaching such points corresponds precisely to the NS limit in the string theory description \cite{Bershtein2021}. It would be interesting to see if retaining higherorder terms in the $\epsilon$-expansion of $W$ along the isomonodromic flows near $\theta=0$ reconstructs further contributions to the NS free energy.
\end{remark}

\bibliographystyle{alph}
\bibliography{JoycePainleve_arxiv1.bib}

\end{document}